\documentclass[lettersize,journal]{IEEEtran}
\usepackage{amsmath,amsfonts}
\usepackage{amsmath,amssymb}
\usepackage{mathtools}
\usepackage{dsfont}
\usepackage{amsthm}
\usepackage{algorithmic}
\usepackage{algorithm}
\usepackage{array}
\usepackage[caption=false,font=normalsize,labelfont=sf,textfont=sf]{subfig}
\usepackage{textcomp}
\usepackage{stfloats}
\usepackage{url}
\usepackage{verbatim}
\usepackage{graphicx}
\usepackage{cite}
\usepackage{siunitx}
\usepackage{empheq}
\newtheorem{thm}{Theorem}
\newtheorem{lem}[thm]{Lemma}
\newtheorem{prop}{Proposition}

\usepackage{ragged2e}
\usepackage{enumitem}
\usepackage{url}
\usepackage{bm}
\usepackage{empheq}
\usepackage[subnum]{cases}
\usepackage{upquote}
\usepackage{multirow}
\allowdisplaybreaks
\DeclareMathOperator{\Tr}{tr}
\DeclareMathOperator{\Diag}{diag}

\begin{document}

\title{Cost-Effective Activity Control of Asymptomatic Carriers in Layered Temporal Social Networks}

\author{Masoumeh~Moradian,
        Aresh~Dadlani,~\IEEEmembership{Senior~Member,~IEEE,}
        Rasul~Kairgeldin,
        and~Ahmad~Khonsari
\thanks{M. Moradian is with the School of Computer Engineering, K. N. Toosi~University of Technology, Tehran, Iran, and also with the School of Computer Science, Institute for Research in Fundamental Sciences (IPM), Tehran, Iran (e-mail: mmoradian@kntu.ac.ir).}
\thanks{A. Khonsari is with the School of Electrical and Computer Engineering, University of Tehran, and also with the School of Computer Science,~Institute for Research in Fundamental Sciences (IPM), Tehran, Iran (e-mail: a\_khonsari@ut.ac.ir).}
\thanks{A. Dadlani is with the Department of Computing Science, University of Alberta, Edmonton T6G 2E8, Canada, and also with the Department of Electrical and Computer Engineering, Nazarbayev University, Astana 010000, Kazakhstan. (e-mail: aresh@ualberta.ca).}
\thanks{R. Kairgeldin is with the Department of Electrical Engineering and Computer Science, University of California, Merced, CA 95343, USA. (e-mail: rkairgeldin@ucmerced.edu).}
}

\markboth{Journal of \LaTeX\ Class Files,~Vol.~14, No.~8, August~2021}%
{Shell \MakeLowercase{\textit{et al.}}: A Sample Article Using IEEEtran.cls for IEEE Journals}


\maketitle

\begin{abstract}
\fontdimen2\font=0.54ex
The robustness of human social networks against~epidemic propagation relies on the propensity for physical contact adaptation. During the early phase of infection, asymptomatic~carriers exhibit the same activity level as susceptible individuals, which presents challenges for incorporating control measures in~epidemic projection models. This paper focuses on modeling~and cost-efficient activity control of susceptible and carrier individuals in the context of the susceptible-carrier-infected-removed (SCIR) epidemic model over a two-layer contact network. In~this model, individuals switch from a static contact layer to create~new links~in a temporal layer based on state-dependent activation rates. We~derive conditions~for the infection to die out or persist in a homogeneous network. Considering the significant costs associated with reducing the activity of susceptible and carrier individuals, we formulate~an~optimization problem to minimize the disease decay rate while constrained~by a limited budget.~We propose the use of successive geometric~programming (SGP) approximation for this optimization task. Through simulation experiments on Poisson random graphs, we assess the impact of different parameters on disease prevalence. The results demonstrate that our SGP framework achieves a cost reduction of nearly $33\%$ compared to conventional methods based on degree and closeness centrality. 
\end{abstract}

\vspace{-0.4em}

\section{Introduction}
\label{sec1}
\fontdimen2\font=0.56ex
\IEEEPARstart{M}{odeling} and analyzing contagion processes over network models exhibiting real-world characteristics has garnered much attention in recent years. Undoubtedly, the global outbreak of COVID-19 has further emphasized the need for~refined epidemic models to assess the effectiveness of preventive policies like quarantine, social~distancing, mask-wearing, and vaccination with greater accuracy \cite{Mishra2021, Hota2021}. Specifically, epidemic models that capture the behavior of asymptomatic but infectious individuals have gained prominence. These individuals, who can transmit the disease without showing symptoms, are indistinguishable from healthy individuals and contribute~to higher disease prevalence, particularly in the early stages of~the epidemic. A notable example is the devastating COVID-19~pandemic, wherein it has been reported that a significant $\SI{23}{\percent}$ of disease transmissions originated from asymptomatic infections~\cite{Liu2020}.

Epidemic models, including those with asymptomatic infectious states, have been widely studied on static graphs to~determine the \emph{epidemic threshold} for disease prevalence \cite{Dadlani2020}. However, human contact networks are inherently dynamic, and the timescale of infection propagation and link evolution is typically comparable. To~address this, various approaches have been proposed to integrate network temporality with disease dynamics \cite{Perra2012,  Valdano2015, Pare2018, Ogura2019, Temirlan2023}. Among these, activity-driven networks (ADNs) \cite{Perra2012} have emerged as an analytically tractable framework for studying~dynamic processes co-evolving over networks by capturing the temporal attributes of individuals. In ADNs, temporal links form when individuals become active and are removed when they are inactive. While social networks naturally exhibit temporal characteristics, the spread of disease and implementation of preventive measures~can alter this attribute \cite{Faryad2019, Zino2020}. Therefore, it is essential to model the temporal aspect as a function of the disease state. For instance, individuals displaying symptoms may limit their routine activities, and government-regulated policies can restrict the activities of healthy or undiagnosed asymptomatic individuals to prevent early-stage infection spread.

In modeling spreading processes over networks, it is crucial to consider both the temporal aspect and the contextual definition of link connections. These connections play a vital role in disease propagation, especially when implementing preventive~policies. This can be exemplified by the intrinsic behavior of infected individuals who tend to limit their social activities while maintaining connections with close family members and friends. Multi-layer network structures provide effective ways to represent different types of links, as each layer in these networks corresponds to a unique content and exhibits a distinct connectivity pattern~\cite{Faryad2013, Kivela2014, Aresh2017}. Temporality can also be incorporated in multi-layered networks by modeling the desired number of layers with ADNs.~In particular, nodes form their connections in ADN-based layers when active and dissolve these connections when inactive.

This paper analyzes the impact of \emph{asymptomatic carriers},~individuals without symptoms but infectious, on the epidemic threshold of \emph{multi-layered temporal networks}, and evaluates~the effectiveness of \emph{activity-reduction control measures} in such~networks. Our work is motivated by the fact that the course of~several epidemics may be expedited by the uncertain proportion of~carrier individuals who do not exhibit symptoms and thus, appear healthy. To study this behavior, we consider the stochastic \emph{susceptible-carrier-infected-removed} (SCIR) epidemic model \cite{Gemmetto2014} over a two-layer temporal network proposed in \cite{Aram2020}, with~the first and second layers modeling the persistent and temporal activity-driven links, respectively. In this model, the activity~pattern of individuals (or nodes) in the second layer depends on their disease state, with nodes behaving similarly in the susceptible and carrier states. Moreover, susceptible individuals become carriers upon contracting the disease and may recover with or without symptoms. The model also enables the examination of control-reduction policies on the social activities of susceptible and carrier individuals, such as lockdowns and reduced working hours, particularly when vaccines are not readily available in the early stages of the infection. Imposing such restrictions~entails substantial economic costs, creating a balance between~controlling the disease spread and the incurred costs. The main contributions of this work can be summarized as follows:
\begin{itemize}
    \item We introduce a mean-field (MF) approximated SCIR model for a two-layer temporal network. We derive the~epidemic threshold for the homogeneous activity network, where the activity rates of nodes depend on~the epidemic state and remain consistent for susceptible and carrier nodes. 
    \item For the non-homogeneous case, we optimize the activity rates of susceptible and carrier individuals to minimize~the disease exponential decay rate while adhering to a given budgetary constraint. Since the susceptible and carrier nodes act alike, the proposed optimization problem admits to a non-convex form and is solved using \emph{successive~geometric programming} (SGP) approximation.
    \item We analyze the impact of varying infection and temporality parameters on infection virality and the effectiveness~of activity control policies. Our numerical results demonstrate the superiority of SGP over the conventional degree and closeness centrality approaches, particularly when the temporal layer is more dense. 
\end{itemize}

The rest of the paper is organized as follows. Section~\ref{sec:literature}~provides a literature review, followed by the proposed coupled~epidemic and temporal network model presented in Section~\ref{sec:system}. Section~\ref{sec:threshold} details the epidemic threshold derivation of the~asymptomatic activity rate in a homogeneous network for both, the original and extended SCIR models. Section~\ref{sec:control} formulates the~cost-aware SGP problem for maximizing the disease decay~rate. Numerical experiment results are discussed in Section~\ref{sec:numerical}.~Finally, Section~\ref{sec:conclusion} concludes the paper.\vspace{-0.4em}

\section{Literature Review}
\label{sec:literature}
\fontdimen2\font=0.50ex
Several related studies have focused on analyzing the~epidemic threshold in temporal networks with different layering combinations. In \cite{Liu2018}, the authors examine the epidemic threshold of the classical SIS model in a network comprising of~two activity-driven layers and partial coupling between nodes. Spectral analysis of the SIS model on a composite~network consisting of static and temporal layers is reported in \cite{Nadini2020}. In~\cite{Aram2020}, the authors derive the epidemic outbreak threshold for a two-layer SIS network model with static and temporal layers while investigating the spread of sexually transmitted diseases (STDs). In their approach, temporal links are formed exclusively between activated nodes and their potential active neighbors, instead of random selection as in ADNs. The work is further~extended by introducing additional epidemic compartments to account for individual protective preferences and the treatment of infected individuals. In \cite{Zhang2020}, the SIS epidemic threshold is derived in temporal~networks including periodic and non-periodic Markovian networks.

The impact of individual awareness on system parameters has been well-investigated in the literature.~In \cite{Hu2018}, the susceptible-alerted-infected-susceptible (SAIS) epidemic model  is studied for a single-layer ADN, where alerted susceptible nodes affect the epidemic threshold. The~authors of \cite{Temirlan2023} propose a novel energy model that relates the SAIS model with the structural balance~of signed networks. By separating the diffusion of awareness from infection spread over distinct layers, the authors of \cite{Wang2021} analyze the coupled dynamics involving self-initiated awareness over multiplex networks. Additionally, awareness has also shown to induce changes in the topological features of networks, as in \cite{Li2021}, where the tendency of aware individuals to participate in specific layers of the physical contact network is explored.

Leveraging effectual control mechanisms in shaping~epidemic outbreaks has also received much attention in recent~years. In \cite{Zino2020}, the authors derive the epidemic threshold for the SAIS model on ADNs, considering awareness campaigns and confinement as control actions. The efficacy of active and inactive~quarantine strategies on epidemic containment for the SIS and the SIR processes on ADNs are extensively analyzed in~\cite{Mancastroppa2020}. The authors of \cite{Ogura2019} obtain an upper bound for the decay rate of the infected population and optimize the activity and acceptance rates of infected nodes, subject to cost and performance constraints. The model from \cite{Ogura2019} is further generalized in \cite{Hota2021b}, where infected individuals are divided into asymptomatic and symptomatic groups, and asymptomatic individuals exhibit activity and acceptance rates similar to susceptible individuals. In \cite{Chen2014}, the authors attempt to concurrently control the spread of multiple processes~by optimally allocating resources to different layers of the directed and weighted layered contact network.


The impact of asymptomatic carriers and their behavioral~similarity to healthy susceptible individuals in temporal networks has been explored in \cite{hota2021impacts} and \cite{9683739}.~Utilizing a one-layer~ADN, these studies assume decentralized activation decisions by nodes based on prevalence information and employ a game-theoretic~approach to model node behavior.~In contrast, our work introduces a novel contact adaptation model within a two-layer network comprising static and temporal layers.~Here, symptomatic individuals voluntarily reduce activities upon infection, while~the activities of carriers and susceptible individuals are subject to~optimal control through external restriction policies, constrained by a~limited budget.~These policies are implemented through measures such as lockdowns, quarantine, and adjusted working hours, resulting in the reduction of activities among individuals.~To the best of our knowledge, this work is the first to model and contain the activity of asymptomatic carriers by employing a temporal model of layered networks and external control with limited budget.\vspace{-0.3em}

\section{Networked Epidemic Model}
\label{sec:system}
\fontdimen2\font=0.50ex
In this section, we first introduce the network~structure for our analysis and then formulate the infection spreading model.\vspace{-0.6em}

\subsection{Temporal Layered Topology}
\label{sub:topology}
\fontdimen2\font=0.50ex
We consider a two-layer network structure, denoted by the ordered tuple $\mathcal{G} \!=\! (V,E_1,E_2)$, where the undirected graphs $G_1 \!=\! (V, E_1)$ and $G_2 \!=\! (V,E_2)$ represent the connections~between the $|V| \!=\! N$ nodes in the first and second layers,~respectively. In each layer $l \!\in\! \{1,2\}$, the link between any two arbitrary nodes $i,j \!\in\! V$ is denoted as $(i,j) \!\in\! E_l$. We use $\mathbf{A}$ and $\mathbf{B}$ to symbolize the adjacency matrices of the graphs $G_1$ and $G_2$, respectively, where the entry $(i,j)$ of $\mathbf{A}$ ($\mathbf{B}$),~denoted by $a_{i,j}$ ($b_{i,j}$), is equal to one if there exists a link between~nodes $i$ and $j$, and zero otherwise. Note that fully connected graphs in both layers are achieved only when $a_{i,j} \!=\! b_{i,j} \!=\! 1$ holds for all $i\neq j$. Nodes can frequently change their activity states from \emph{inactive} to \emph{active} and vice versa. In either state, the nodes maintain connections with their neighbors in the first layer.~In the second layer, links are potential links and are generated with a given probability only when both end nodes are in an \emph{active} state. More specifically, the link $(i,j) \!\in\! E_2$ is created with probability $p_{i,j}$ if nodes $i$ and $j$ are both active. This temporal link persists until either node $i$ or node $j$ becomes inactive. Thus, the links in $G_1$ remain fixed regardless of~the activity states of the nodes, whereas $G_2$ represents the links~that are probabilistically created when both end nodes are active.\vspace{-0.9em}

\subsection{Epidemic Model Description}
\label{sub:epidemic}
\fontdimen2\font=0.50ex
We adopt the paradigmatic SCIR compartmental model \cite{Gemmetto2014}, in which each individual $i \in V$ is in one of the following four epidemic states at any given time: \emph{susceptible} to the disease ($\textrm{S}$), an \emph{asymptomatic carrier} ($\textrm{C}$), \emph{infected} by the disease ($\textrm{I}$), or \emph{removed} ($\textrm{R}$) after recovery. That is, upon contracting the~disease, a susceptible person either becomes a carrier of the disease (i.e., infectious without showing signs of any symptoms) or transitions directly to the infected state, displaying visible symptoms. The carrier, in turn, either becomes infectious before recovering from the illness or recovers directly without exhibiting any symptoms.~A recovered individual does not return to the susceptible state again.

Extending the SCIR model to both layers of the network $\mathcal{G}$ results in eight possible epidemic states,~namely \emph{susceptible-inactive} ($\textrm{S}^1$), \emph{carrier-inactive} ($\textrm{C}^1$), \emph{infected-inactive} ($\textrm{I}^1$), \emph{removed-inactive} ($\textrm{R}^1$), \emph{susceptible-active} ($\textrm{S}^2$), \emph{carrier-active} ($\textrm{C}^2$), \emph{infected-active} ($\textrm{I}^2$), and \emph{removed-active}($\textrm{R}^2$). All processes, including the duration of active, inactive, carrier, and infected states, as well as the disease transmission process are assumed to be exponentially distributed. As a result of this, the network state at time $t$ can be modeled as an $N$-tuple continuous-time Markov chain (CTMC), expressed as $\textbf{X}(t)= \{(X_i(t)); t \geq 0\}$, where $i \!=\! 1,2,\ldots, N$. For $l \in \{1,2\}$, $X_i(t) \in \{\textrm{S}^l,\textrm{C}^l,\textrm{I}^l,\textrm{R}^l\}$, denotes the state of node $i$ at time $t$.
\begin{figure}[!t]
	\centering
	\includegraphics[width=0.9\columnwidth]{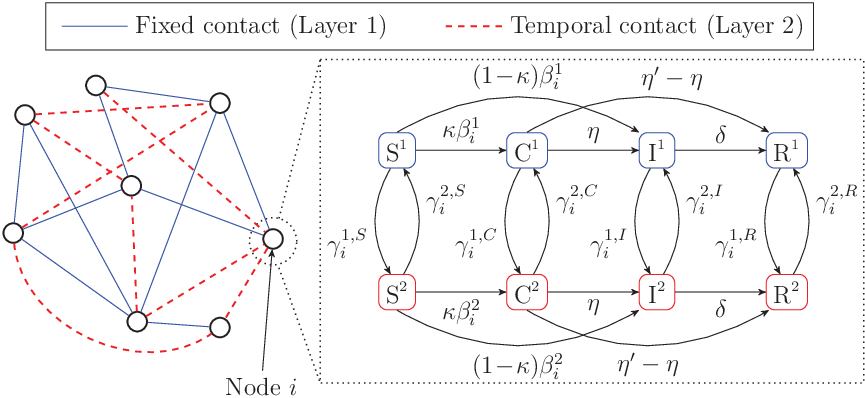}
	\vspace{-0.8em}
 	\caption{The state transitions for the SCIR model over a two-layer network.}
	\label{fig_1}
 \vspace{-0.3em}
\end{figure}

\subsubsection{Transition Rates of the CTMC Model}
\label{sec:model}
\fontdimen2\font=0.50ex
We now introduce the possible state transitions for node $i$~as illustrated in \figurename~\ref{fig_1}. The sojourn times of inactive and active states depend on the epidemic status of the node. Specifically, node $i$ remains inactive (active) for an exponential random time with a rate of $\gamma_i^{1,\textrm{X}} \!>\! 0$ ($\gamma_i^{2,\textrm{X}} \!>\! 0$), where $\textrm{X} \!\in\! \{\textrm{S}, \textrm{C}, \textrm{I}, \textrm{R}\}$ denotes the epidemic state of node $i$. Consequently, the~transition rates between states $\textrm{X}^1$ and $\textrm{X}^2$,
for a chosen time-step $\Delta t$, can be written as \cite{Diekmann2009}:\vspace{-0.2em}
\begin{equation}
   \begin{cases}
    \!\Pr\!\left(X_i(t+\Delta t) =\textrm{X}^2|X_i(t)=\textrm{X}^1,\textbf{X}(t)\right)=\gamma_i^{1,\textrm{X}}  \Delta t \!+\! o(\Delta t), \vspace{0.1em}\\
    \!\Pr\!\left(X_i(t+\Delta t)=\textrm{X}^1|X_i(t)=\textrm{X}^2,\textbf{X}(t)\right)=\gamma_i^{2,\textrm{X}}  \Delta t \!+\! o(\Delta t).
  \end{cases}
  \label{eq0}
  \vspace{-0.2em}
\end{equation}

We define $\beta_{\,\textrm{C}} > 0$ and $\beta_{\,\textrm{I}} > 0$ to be the infection~transmission rates from a carrier and an infected neighbor, respectively. Thus, the rate at which a susceptible-inactive node ($\textrm{S}^1$) contracts the disease is given by:\vspace{-0.2em}
\begin{equation}
\beta_i^1(t) \triangleq \beta_{\,\textrm{C}} Z_i(t) + \beta_{\,\textrm{I}} Y_i(t),
\label{eq:beta1}
\vspace{-0.2em}
\end{equation}
where $Z_i(t)$ and $Y_i(t)$ are, respectively, the number of carrier and infected neighbors of node $i$ in the (fixed) first layer. As a result of this, we have:\vspace{-0.2em}
\begin{equation}
   \begin{cases}
    \!Z_i(t) = \sum_{k \in V} a_{i,k} \mathds{1}_{\{X_k(t) \in \{\textrm{C}^1,\textrm{C}^2\}\}},\vspace{0.2em} \\
    \!Y_i(t) = \sum_{k \in V} a_{i,k} \mathds{1}_{\{X_k(t) \in \{\textrm{I}^1,\textrm{I}^2\}\}},
  \end{cases}
  \label{eq:ZY}
  \vspace{-0.2em}
\end{equation}
where $\mathds{1}_{\{\cdot\}}$ denotes the identity operator.
Moreover, regardless~of being active or inactive, a susceptible node transitions to the~carrier and infected states with probabilities $\kappa$ and $\bar{\kappa} \!=\! 1 - \kappa$~upon contracting the infection, respectively. Therefore, we have:\vspace{-0.2em}
\begin{equation}
   \begin{cases}
    \!\Pr\!\left(X_i(t+\Delta t) \!=\! \textrm{C}^1|X_i(t) \!=\! \textrm{S}^1,\textbf{X}(t)\right) \!=\! \kappa \beta^1_i(t) \Delta t \!+\! o(\Delta t), \vspace{0.2em}\\
    \!\Pr\!\left(X_i(t+\Delta t) \!=\! \textrm{I}^1|X_i(t) \!=\! \textrm{S}^1,\textbf{X}(t)\right) \!=\! \bar{\kappa}\beta^1_i(t) \Delta t \!+\! o(\Delta t).
  \end{cases}
  \label{eq:rateSus}
\end{equation}

Furthermore, a susceptible-active node ($\textrm{S}^2$) becomes infectious with non-negative rate:\vspace{-0.2em}
\begin{equation}
 \beta_i^2(t) \triangleq \beta_i^1(t) + \beta_i^a(t),
\label{eq:beta2}
\end{equation}
where $\beta_i^1(t)$ (derived in \eqref{eq:beta1}) and $\beta_i^a(t)$ are the infection rates associated with the links of the first (static) and second (temporal) layers, respectively. Here, $\beta_i^a(t) \triangleq \beta_{\,\textrm{C}} Z_i^a(t) + \beta_{\,\textrm{I}} Y_i^a(t)$, where $Z_i^a(t)$ and $Y_i^a(t)$ are the number of active carrier and infected neighbors of node $i$ in the second layer that have an activated link with node $i$, respectively. Accordingly,\vspace{-0.2em}
\begin{equation}
   \begin{cases}
    \!Z_i^a(t) = \sum_{k \in V} b_{i,k}\,\zeta_{i,k} \mathds{1}_{\{X_k(t) = \textrm{C}^2\}},\vspace{0.2em} \\
    \!Y_i^a(t) \!=\! \sum_{k \in V} b_{i,k}\,\zeta_{i,k} \mathds{1}_{\{X_k(t) = \textrm{I}^2\}},
  \end{cases}
  \label{eq:ZYa}
\end{equation}
where $\zeta_{i,k}$ is a Bernoulli random variable that takes the value of one with probability $p_{i,k}$, given that nodes $i$ and $k$ are both active. Subsequently, the equations corresponding to a susceptible active node are obtained as follows:
\vspace{0em}
\begin{equation}
   \begin{cases}
    \!\Pr\!\left(X_i(t+\Delta t) \!=\! \textrm{C}^2|X_i(t) \!=\! \textrm{S}^2,\textbf{X}(t)\right) \!=\! \kappa \beta^2_i(t)  \Delta t \!+\! o(\Delta t), \vspace{0.2em}\\
    \!\Pr\!\left(X_i(t+\Delta t) \!=\! \textrm{I}^2|X_i(t) \!=\! \textrm{S}^2,\textbf{X}(t)\right) \!=\! \bar{\kappa}\beta^2_i(t) \Delta t \!+\! o(\Delta t).
  \end{cases}
  \label{eq:rateSusAct}
\end{equation}

Suppose that a carrier node, whether active or inactive, becomes infected at rate $\eta > 0$ and recovers at a rate of $\eta' - \eta > 0$. Furthermore, an infected node recovers at a rate of $\delta > 0$. Thus,\vspace{-0.2em}
\begin{equation}
   \begin{cases}
    \!\Pr\!\big(X_i(t \!+\! \Delta t) \!=\! \textrm{I}^l | X_i(t) \!=\! \textrm{C}^l, \textbf{X}(t)\big) \!=\! \eta \Delta t + o(\Delta t), \vspace{0.2em}\\
    \!\Pr\!\big(X_i(t \!+\! \Delta t) \!=\! \textrm{R}^{l} | X_i(t) \!=\! \textrm{C}^{l}, \textbf{X}(t)\big) \!=\! (\eta' \!-\! \eta )\Delta t \!+\! o(\Delta t), \vspace{0.2em}\\
    \!\Pr\!\big(X_i(t + \Delta t) \!=\! \textrm{R}^l | X_i(t) \!=\! \textrm{I}^l, \textbf{X}(t)\big) \!=\! \delta \Delta t + o(\Delta t).
  \end{cases}
  \label{eq:ratem}
\end{equation}
where $l \!\in\! \{1,2\}$. The networked Markov process introduced~has a state space of $8^N$ states. As such, deriving the probability~distribution of each node being in one of the epidemic states~over time for this Markov chain, especially for large networks, is~mathematically intractable due to the state space explosion problem.~To overcome this, we employ a first-order MF approximation.\vspace{-0.3em}

\subsubsection{Mean-field Approximation Model}
\label{sec:mfa}
\fontdimen2\font=0.50ex
To reduce the dimensionality of the proposed CTMC model, we employ the approximation method for MF differential equations~\cite{Hota2021,Faryad2013,Aresh2017}. For this purpose, we first define the following state probabilities:\vspace{-0.2em}
\begin{equation}
   \begin{cases}
    S_i^l(t) \triangleq \Pr\!\big(X_i(t) = \textrm{S}^l\big), \vspace{0.2em}\\
    C_i^l(t) \triangleq \Pr\!\big(X_i(t) = \textrm{C}^l\big), \vspace{0.2em}\\
    I_i^l(t) \triangleq \Pr\!\big(X_i(t) = \textrm{I}^l\big), \vspace{0.2em}\\
    R_i^l(t) \triangleq \Pr\!\big(X_i(t) = \textrm{R}^l\big),
  \end{cases}
  \label{eq:stateProb}
\end{equation}
where $l \!\in\! \{1, 2\}$ and $\sum_i\sum_l S_i^l(t)+C_i^l(t)+I_i^l(t)+R_i^l(t)=1$. On the other hand, from \eqref{eq:beta1} and \eqref{eq:ZY}, the expected value of $\beta^1_i(t)$ can be expressed as follows:\vspace{-0.2em}
\begin{equation}
\begin{aligned}
	\mathbb{E}\!\left[\beta^1_i(t)\right] \!=\!  &\sum_{k=1}^N \!a_{i,k} \Big(\beta_{\,\textrm{C}} \Pr\!\big(X_k(t) \in \{\textrm{C}^1,\textrm{C}^2\}|X_i(t) = \textrm{S}^1\big)  \\
\! &+ \beta_{\,\textrm{I}} \Pr\!\big(X_k(t) \in \{\textrm{I}^1,\textrm{I}^2\}|X_i(t) = \textrm{S}^1\big)\!\Big).
\label{eq:beta1e}
\end{aligned}
\end{equation}
Note that the condition $X_i(t) \!=\! \textrm{S}^1$ given in \eqref{eq:beta1e} arises directly from \eqref{eq:rateSus}, where $\beta^1_i(t)$ is the disease contraction rate of node $i$ given that it is susceptible-inactive. In the MF  model, the states of the nodes are assumed to be uncorrelated. Thus, we have $\Pr\!\big(X_k(t) \!\in\! \{\textrm{C}^1,\textrm{C}^2\}|X_i(t) \!=\! \textrm{S}^1\big) \!=\! \Pr\!\big(X_k(t) \!\in\! \{\textrm{C}^1,\textrm{C}^2\}\big)$ and $\Pr\!\big(X_k(t) \!\in\! \{\textrm{I}^1,\textrm{I}^2\}|X_i(t) \!=\! \textrm{S}^1\big) \!=\! \Pr\!\big(X_k(t) \!\in\! \{\textrm{I}^1,\textrm{I}^2\}\big)$. Using the state probabilities in \eqref{eq:stateProb},~we can now rewrite \eqref{eq:beta1e} as:\vspace{-0.2em}
\begin{equation}
	\mathbb{E}\!\left[\beta^1_i(t)\right] \!=\!  \sum_{k=1}^N \!a_{i,k} \Big(\!\beta_{\,\textrm{C}} \big(C^1_k(t) \!+\! C^2_k(t)\big) \!+\! \beta_{\,\textrm{I}} \big(I^1_k(t) \!+\! I^2_k(t)\big)\!\Big).
\label{eq:beta1e_new}
\end{equation}

Similarly, using \eqref{eq:beta2}, \eqref{eq:ZYa}, and \eqref{eq:beta1e}, the expected value of $\beta^2_i(t)$ is derived as follows, where $p_{i,j} = \mathbb{E}[\zeta_{i,j}]$:
\begin{equation}
	\begin{aligned}
		\mathbb{E}\!\left[\beta^2_i(t)\right] \!=\! & \sum_{k=1}^N \!a_{i,k} \Big(\!\beta_{\,\textrm{C}} \big(C^1_k(t) + C^2_k(t)\big) \!+\! \beta_{\,\textrm{I}} \big(I^1_k(t)+I^2_k(t)\big)\!\Big)\\
		&+ \sum_{k=1}^N b_{i,k}\, p_{i,k} \big(\beta_{\,\textrm{C}} C^2_k(t) \!+\! \beta_{\,\textrm{I}} I^2_k(t)\big).
	\end{aligned}
\label{eq:beta2e}
\end{equation}

Based on \eqref{eq:stateProb}, \eqref{eq:beta1e_new}, and \eqref{eq:beta2e}, the system of MF~equations for the original CTMC can be characterized as follows:\vspace{-0.2em}
\begin{equation}
    \begin{cases}
    \begin{aligned}
        {S_i^1}^{\prime\!}(t)\! & = -\Big(\!\gamma_i^{1,\textrm{S}\!} + \mathbb{E}\!\left[\beta^1_i(t)\right]\!\Big) S^1_i(t) \!+\! \gamma^{2,\textrm{S}}_{i}\, S_i^2(t), \\
        {S_i^2}^{\prime\!}(t)\! & = -\Big(\!\gamma_i^{2,\textrm{S}\!} + \mathbb{E}\!\left[\beta^2_i(t)\right]\!\Big) S^2_i(t) \!+\! \gamma^{1,\textrm{S}}_{i}\, S_i^1(t),\\
        {C_i^1}^{\prime\!}(t)\! & = -\Big(\!\gamma_i^{1,\textrm{C}\!} \!+\! \eta'\Big) C_i^1(t) \!+\! \gamma_i^{2,\textrm{C}\!}\,C_i^2(t) \!+\! \kappa\, \mathbb{E}\!\left[\beta^1_i(t)\right]\!S_i^1(t), \\
        {C_i^2}^{\prime\!}(t)\! & = -\Big(\!\gamma_i^{2,\textrm{C}\!} \!+\! \eta'\Big) C_i^2(t) \!+\! \gamma_i^{1,\textrm{C}\!}\,C_i^1(t) \!+\! \kappa\, \mathbb{E}\!\left[\beta^2_i(t)\right]\!S_i^2(t), \\
        {I_i^1}^{\prime\!}(t)\! & = -\Big(\!\gamma_i^{1,\textrm{I}\!} \!+\! \delta\Big) I_i^1(t) \!+\! \gamma_i^{2,\textrm{I}}\,I_i^2(t) + \eta\,C_i^1(t) \\
        &\quad\, + \bar{\kappa}\, \mathbb{E}\!\left[\beta^1_i(t)\right] S_i^1(t),\\
        {I_i^2}^{\prime\!}(t)\! & = -\Big(\!\gamma_i^{2,\textrm{I}\!} \!+\! \delta\Big) I_i^2(t) \!+\! \gamma_i^{1,\textrm{I}}\,I_i^1(t) + \eta\,C_i^2(t) \\
        &\quad\, + \bar{\kappa}\, \mathbb{E}\!\left[\beta^2_i(t)\right] S_i^2(t), \\
        {R_i^1}^{\prime\!}(t)\! & = -\gamma_i^{1,\textrm{R}\!} R_i^1(t) \!+\! \gamma_i^{2,\textrm{R}\!} R_i^2(t) \!+\! \delta\,I_i^1(t) \!+\! (\eta' \!-\! \eta)\,C_i^1(t), \\
        {R_i^2}^{\prime\!}(t)\! & = -\gamma_i^{2,\textrm{R}\!} R_i^2(t) \!+\! \gamma_i^{1,\textrm{R}\!} R_i^1(t) \!+\! \delta\,I_i^2(t) \!+\! (\eta' \!-\! \eta)\,C_i^2(t).
    \end{aligned}
    \end{cases}
    \label{eq:mfCTMC}
\end{equation}

The MF epidemic model in \eqref{eq:mfCTMC} has a reduced~dimension~of $8N$ (more details in \cite{Taylor2011}). Note that the carrier nodes are unaware of their infectiousness and thus, they maintain the same activity level as their susceptible peers, i.e., $\forall i \!\in\! V, \gamma_i^{1,\textrm{C}} \!=\! \gamma_i^{1,\textrm{S}}$.~Assuming the absence of a vaccine,~we further assume that to curb the prevalence of the disease, limitations are enforced on the activity rates of all susceptible and carrier individuals as they are indistinguishable. While limiting the activity level of these individuals inflicts economic costs on society, relaxing the restrictions contributes to the spread of the disease by~the~asymptomatic carriers.\vspace{-0.3em}

\section{Epidemic Threshold Analysis}
\label{sec:threshold}
\fontdimen2\font=0.50ex
The focus of this section is on large homogeneous networks, where the activity and inactivity rates of all nodes in epidemic state $\textrm{X} \!\in\! \{\textrm{S}, \textrm{C}, \textrm{I}, \textrm{R}\}$ are homogeneous. More precisely, $\forall i \!\in\! V, \gamma_i^{l,\textrm{X}} \!=\! \gamma^{l,\textrm{X}}$, where $l \!\in\! \{1, 2\}$. In addition, we model $G_1$ and $G_2$ as random regular graphs with degrees $d_1$ and $d_2$, respectively. Also, the probability of constructing a temporal link in the second layer is taken to be $p$ for all nodes. Finally, since susceptible and carrier nodes have the same activity pattern, we set $\gamma^{1,\textrm{S}} \!=\! \gamma^{1,\textrm{C}} \!=\! \gamma^1$ and $\gamma^{2,\textrm{S}} \!=\! \gamma^{2,\textrm{C}} \!=\! \gamma^2$. Due to network homogeneity, the probabilities associated with different epidemic states will be the same for all nodes. By dropping the subscript $i$ and applying the above~assumptions, the model in \eqref{eq:mfCTMC} can be formulated as:\vspace{-0.2em}
\begin{subnumcases}{\label{eq:MFhomo}}
	\!{S^1}^{\prime\!}(t)\!=\!&$\!\!\!\!-\Big(\gamma^1 + \mathbb{E}\!\left[\beta^1(t)\right]\!\Big)\,S^1(t) \!+\! \gamma^2\, S^2(t),$ \label{eq:MFHomoS1} \\
	\!{S^2}^{\prime\!}(t)\!=\!&$\!\!\!\!-\Big(\gamma^2 + \mathbb{E}\!\left[\beta^2(t)\right]\!\Big)\,S^2(t) \!+\! \gamma^1\, S^1(t),$ \label{eq:MFHomoS2} \\
	\!{C^1}^{\prime\!}(t)\!=\!&$\!\!\!\!-\Big(\!\gamma^{1\!} \!+\! \eta'\!\Big) C^{1\!}(t) \!+\! \gamma^2 C^2(t) \!+\! \kappa\, \mathbb{E}\!\left[\beta^{1\!}(t)\right]\! S^{1\!}(t),$ \label{eq:MFHomoC1} \\	
	\!{C^2}^{\prime\!}(t)\!=\!&$\!\!\!\!-\Big(\!\gamma^{2\!} \!+\! \eta'\!\Big) C^{2\!}(t) \!+\! \gamma^1 C^1(t) \!+\! \kappa\, \mathbb{E}\!\left[\beta^{2\!}(t)\right]\! S^{2\!}(t),$ \label{eq:MFHomoC2} \\
	\!{I^1}^{\prime\!}(t)\!=\!&$\!\!\!\!-\Big(\!\gamma^{1\!} + \delta\Big) I^{1}(t) + \gamma^{2} I^2(t) + \eta\, C^1(t)$ \nonumber \vspace{-0.1em}\\
	&$\!\!\!+ \bar{\kappa}\, \mathbb{E}\!\left[\beta^1(t)\right] S^1(t),$ \label{eq:MFHomoI1} \\
	\!{I^2}^{\prime\!}(t)\!=\!&$\!\!\!\!-\Big(\!\gamma^{2\!} + \delta\Big) I^2(t) + \gamma^{1} I^1(t) + \eta\, C^2(t)$ \nonumber \vspace{-0.1em}\\
	&$\!\!\!+ \bar{\kappa}\, \mathbb{E}\!\left[\beta^2(t)\right] S^2(t),$ \label{eq:MFHomoI2} \\
	\!{R^1}^{\prime\!}(t)\!=\!&$\!\!\!\!-\gamma^{1\!} R^{1\!}(t) \!+\! \gamma^{2} R^2(t) \!+\! \delta I^{1\!}(t) \!+\! (\eta' \!-\! \eta) C^{1\!}(t),$ \label{eq:MFHomoR1}\\
	\!{R^2}^{\prime\!}(t)\!=\!&$\!\!\!\!-\gamma^{2\!} R^{2}(t) \!+\! \gamma^{1\!} R^1(t) \!+\! \delta I^{2\!}(t) \!+\! (\eta' \!-\! \eta) C^2(t).$ \label{eq:MFHomoR2}
 \vspace{-0.2em}
\end{subnumcases}
where $\mathbb{E}[\beta^1(t)]$ and $\mathbb{E}[\beta^2(t)]$, derived respectively in \eqref{eq:beta1e_new} and \eqref{eq:beta2e}, reduce to the following for the homogeneous case:\vspace{-0.2em}
\begin{subnumcases}{\label{eq:15}}
	\!\!\mathbb{E}\!\left[\beta^{1\!}(t)\right] \!=\! d_{1\!} \Big(\!\beta_{\,\textrm{C}} \big(C^{1\!}(t) \!+\! C^2(t)\big) \!+\! \beta_{\,\textrm{I}} \big(I^{1\!}(t) \!+\! I^2(t)\big)\!\Big), \label{eq:beta1e_homo} \\
	\!\!\mathbb{E}\!\left[\beta^{2\!}(t)\right] \!=\! d_{1\!} \Big(\!\beta_{\,\textrm{C}} \big(C^1(t) \!+\! C^2(t)\big) \!+\! \beta_{\,\textrm{I}} \big(I^1(t) + I^2(t)\big)\!\Big) \nonumber \\
    \qquad\qquad\,\, + d_2\, p\, \big(\beta_{\,\textrm{C}} C^2(t) + \beta_{\,\textrm{I}} I^2(t)\big). \label{eq:beta2e_homo}
\end{subnumcases}
\vspace{-0.8em}

\subsection{Original SCIR Model}
\label{sec:org_scir}
\fontdimen2\font=0.50ex
For the case when $\kappa \!=\! 1$, susceptible nodes that come in~contact with carrier or infected neighbors become asymptomatic before progressing to the infected epidemic state. We refer to this scenario as the \emph{original SCIR model}. By setting $\kappa \!=\! 1$~in \eqref{eq:MFhomo},~the dynamical system can be analyzed in terms of the \emph{basic~reproduction number} $\mathcal{R}_0$, which refers to the average number of secondary infections caused by~a~single carrier introduced into a susceptible network, and essentially determines whether the disease prevails ($\mathcal{R}_0 \!>\! 1$) or eventually dies out ($\mathcal{R}_0 \!<\! 1$) in the network~\cite{Dadlani2020}. 

To investigate the system stability, we now analyze the steady-state behavior of the non-linear model given in \eqref{eq:MFhomo}, regardless of the initial fraction of infected nodes. The trivial \emph{disease-free equilibrium} (DFE) of \eqref{eq:MFhomo}, denoted by $\mathcal{E}_0$, is:
\begin{align}
	\mathcal{E}_0 &= \left( S^1, S^2, C^1, C^2, I^1, I^2, R^1, R^2 \right) \nonumber \\
	    		  &= \left( \frac{\gamma^{2}}{\gamma^{1} \!+\! \gamma^{2}}, \frac{\gamma^{1}}{\gamma^{1} \!+\! \gamma^{2}},0,0,0,0,0,0 \right).
\end{align}
The equilibrium $\mathcal{E}_0$ is locally stable if and only if $\mathcal{R}_0 \!<\! 1$.~By undertaking the \emph{next generation matrix} (NGM) approach \cite{Diekmann2009}, the basic reproduction number can be obtained as $\mathcal{R}_0 \!=\! \rho(\mathbf{F \cdot V}^{-1})$, where $\rho(\cdot)$ denotes the spectral radius.~Moreover, the order of the square matrices $\mathbf{F}$ and $\mathbf{V}$ is $m$, which essentially signifies the number of disease-infected compartments. In our model, $m \!=\! 4$ as there are four infection-related compartments, namely $\textrm{C}^1, \textrm{C}^2, \textrm{I}^1,$ and $\textrm{I}^2$. Likewise,~the $(i,j)$-th entry of $\mathbf{F}$ is derived as $\mathbf{F}_{i,j} \!=\! \partial \mathcal{F}_i/\partial x_j$ evaluated at~$\mathcal{E}_0$, where $\mathcal{F}_i$ is the rate of appearance of new infections in compartment $i$ and $x_j$ denotes the probability of being in compartment $j$. Here, we let $x_1 \!=\! C^1$, $x_2 \!=\! C^2$, $x_3 \!=\! I^1$, and $x_4 \!=\! I^2$. In our model, susceptible nodes transition to either carrier or infected states when new infections occur. Thus, the values of $\mathcal{F}_1$, $\mathcal{F}_2$, $\mathcal{F}_3$, and $\mathcal{F}_4$ are the same as the last~terms in \eqref{eq:MFHomoC1}, \eqref{eq:MFHomoC2}, \eqref{eq:MFHomoI1}, and \eqref{eq:MFHomoI2}, respectively. However, in the original SCIR model, $\mathcal{F}_3 = \mathcal{F}_4 = 0$ since $\bar{\kappa}=0$.~Finally, the $(i,j)$-th entry of $\mathbf{V}$ is derived as $\mathbf{V}_{i,j} \!=\! \partial \mathcal{V}_i/\partial x_j$ evaluated at $\mathcal{E}_0$, where $\mathcal{V}_i \!=\! \mathcal{V}_i^- \!+\! \mathcal{V}_i^+$, and $\mathcal{V}_i^-$ and $\mathcal{V}_i^+$~denote, respectively, the outgoing and incoming rates into compartment $i$ through all means other than new infections. For instance, $\mathcal{V}_1^- \!=\! (\gamma^{1,\textrm{C}} \!+\! \eta') C^1$ and $\mathcal{V}_1^+ \!=\! \gamma^{2,\textrm{C}} C^2$. Accounting for all the above definitions, we postulate the following lemma.
\begin{lem}
	\label{lemma1}
	For the original SCIR model, the basic~reproduction number is $\mathcal{R}_0 \!=\! \rho\left(\mathbf{L}\right) \!=\! \rho\big(\mathbf{F_1}\cdot (\beta_{\,\textrm{C}} \mathbf{I} + \eta\, \beta_{\,\textrm{I}} \mathbf{V}_2^{-1})\cdot \mathbf{V}_1^{-1}\big)$, where:\vspace{-0.2em}
	\begin{align}
		\mathbf{F}_1 &=
					\begin{bmatrix}
						d_1\, S^1 & d_1\, S^1\\
						d_1\, S^2 & (p\,d_2 + d_1)\, S^2
					\end{bmatrix},  \nonumber \\				
		\mathbf{V}_1 &=
					\begin{bmatrix}
						\gamma^{1} + \eta' & -\gamma^{2}\\
						-\gamma^{1} & \gamma^{2} + \eta'
					\end{bmatrix},~  \text{and} \nonumber \\
		\mathbf{V}_2 &=
					\begin{bmatrix}
						\gamma^{1,\textrm{I}} + \delta & -\gamma^{2,\textrm{I}}\\
						-\gamma^{1,\textrm{I}} & \gamma^{2,\textrm{I}} + \delta
					\end{bmatrix}.  \nonumber
	\end{align}
\end{lem}
\begin{proof}
	See Appendix~\ref{app_A}.
\end{proof}

We now define two thresholds, namely $R_0^{(1)}$ and $R_0^{(2)}$, as:\vspace{-0.3em}
\begin{subnumcases}{\label{eq:R012}}
	\!R_0^{(1)} =&$\!\!\!\!d_1 \left(\cfrac{\beta_{\,\textrm{C}}}{\eta'} + \cfrac{\eta\, \beta_{\,\textrm{I}}}{\eta' \delta}\right),$ \label{eq:R01} \vspace{0.2em}\\
	\!R_0^{(2)} =&$\!\!\!\!p\,d_2 \left(\cfrac{\beta_{\,\textrm{C}}}{\eta'} + \cfrac{\eta \left(\delta \!+\! \gamma^{1,\textrm{I}}\right) \beta_{\,\textrm{I}}}{\eta' \delta \left(\delta \!+\! \gamma^{1,\textrm{I}} \!+\! \gamma^{2,\textrm{I}}\right)}\right).$ \label{eq:R02}
    \vspace{-0.3em}
\end{subnumcases}
Theoretically, the threshold $R_0^{(1)}$ corresponds to the secondary number of infections caused by a carrier node over graph $G_1$ as depicted in \figurename{~\ref{fig_x}a}. In fact, the term $1/\eta'$ ($1/\delta$) in \eqref{eq:R01} is the average time spent by a node in carrier (infected) state. Accordingly, the term $\beta_{\,\textrm{C}}/\eta'$ ($\beta_{\,\textrm{I}}/\delta$) in \eqref{eq:R01} is the~average number of infections produced per link by a carrier (infected) node. Also, a carrier node proceeds to~state $\textrm{I}$ with probability $\eta/\eta'$. Similarly, $R_0^{(2)}$ is the epidemic threshold of the model given in \figurename{~\ref{fig_x}b} over graph $G_2$, where each link in $G_2$ exists with probability $p$. Within this particular framework, infected nodes isolate (do not isolate) with rate $\gamma^{2,\textrm{I}}\, (\gamma^{1,\textrm{I}})$. However, note that those in isolation (i.e., in state $\textrm{I}_{\text{iso}}$) do not transmit the infection to others. The equation \eqref{eq:R02} includes the term $\left(\delta \!+\! \gamma^{1,\textrm{I}}\right)/\left(\delta(\delta \!+\! \gamma^{1,\textrm{I}} \!+\! \gamma^{2,\textrm{I}})\right)$, which represents the average residual time of an infected node in state $\textrm{I}$ shown in \figurename{~\ref{fig_x}b} before transitioning to the recovered state. The subsequent proposition introduces the conditions on the activity rates of carrier nodes under which the disease dies~out or persists in the network.\vspace{-0.2em}
\begin{prop}
	\label{prop1}
	Given the definitions of $R_0^{(1)}$ and $R_0^{(2)}$ in \eqref{eq:R012}, the stability of the DFE can be characterized as follows:
	\begin{itemize}[leftmargin=*,labelindent=2.5mm,labelsep=3.0mm]
	\setlength\itemsep{0.2em}
		\item Case I: If $R_0^{(1)} \!+\! R_0^{(2)} \!<\! 1$, the DFE is stable for all $\gamma^1$.
		\item Case II: If $R_0^{(1)} \!>\! 1$, then the DFE is unstable for all $\gamma^1$.
		\item Case III: If $R_0^{(1)} \!<\! 1$ and $R_0^{(1)} \!+\! R_0^{(2)} \!>\! 1$, then there exists some $\gamma^{1\ast} > 0$, where for $\gamma^1 < \gamma^{1\ast}$, the DFE is stable.
	\end{itemize}
\end{prop}
\begin{proof}
	See Appendix~\ref{app_B}.
    \vspace{-0.4em}
\end{proof}
\begin{figure}[!t]
	\centering
	\includegraphics[width=0.72\columnwidth]{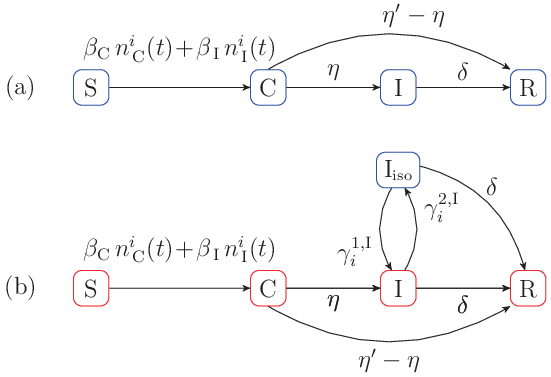}
	\vspace{-0.6em}
 	\caption{Epidemic models corresponding to (a) $R^{(1)}_0$ and (b) $R^{(2)}_0$ thresholds. The number of carrier and infected neighbors of node $i$ are $n^i_{\,\textrm{C}}(t)$ and $n^i_{\,\textrm{I}}(t)$, respectively.}
	\label{fig_x}
 \vspace{-0.4em}
\end{figure}

Building on the intuitions provided for $R_0^{(1)}$ and $R_0^{(2)}$,~it can~be shown that $R_0^{(1)} \!+\! R_0^{(2)}$ marks the epidemic threshold for the network when $\gamma^2 \!=\! 0$, implying that the susceptible and carrier nodes remain active indefinitely upon activation, thus having the most significant impact on disease propagation. Therefore, if the epidemic dissipates under such~circumstances, i.e., if $R_0^{(1)} \!+\! R_0^{(2)} \!<\! 1$ (as shown in Case~I~of Proposition~\ref{prop1}), the DFE will remain stable for all $\gamma^1$ values in the original network.\vspace{-0.6em}

\subsection{Extended SCIR Model}
\label{sec:ext_scir}
\fontdimen2\font=0.50ex
Unlike the preceding sub-section, where susceptible nodes~become carriers upon contracting the infection (i.e., $\kappa \!=\! 1$), we now proceed with the setting when $\kappa \!\neq\! 1$. The corresponding epidemic threshold of the extended SCIR network model is derived in the following lemma.
\begin{lem}
\label{lemma2}
	For the extended SCIR model, the basic reproduction number is $\mathcal{R}_0 \!=\! \rho \left(\kappa \mathbf{L}+ \bar{\kappa} \mathbf{U}\right)$, where $\mathbf{U} \!=\! \beta_{\,\textrm{I}} \mathbf{F}_1\cdot \mathbf{V}^{-1}_2$~and matrices $\mathbf{L}$, $\mathbf{F}_1$, and $\mathbf{V}_2$ are the same as in Lemma~\ref{lemma1}. 
    \vspace{-0.3em}
\end{lem}
\begin{proof}
	See Appendix~\ref{app_C}.
 \vspace{-0.2em}
\end{proof}

The proposition below delineates conditions for epidemic~outbreak in the extended model by defining two new thresholds in terms of ${R}_0^{(1)}$ and ${R}_0^{(2)}$ given in \eqref{eq:R012}.
\begin{prop}
	\label{prop2}
	By defining the following two thresholds in terms of ${R}_0^{(1)}$ and ${R}_0^{(2)}$:\vspace{-0.3em}
	\begin{equation}
  		\begin{cases}
    		\tilde{R}_0^{(1)}\!\!\!\! &= \kappa\, {R}_0^{(1)} + \bar{\kappa}\, \cfrac{d_1 \beta_{\,\textrm{I}}}{\delta}\,,  \vspace{0.2em} \\
    
    		\tilde{R}_0^{(2)}\!\!\!\! &= \kappa\, {R}_0^{(2)} + \bar{\kappa}\, \cfrac{p \, d_2 \left(\delta+\gamma^{1,\textrm{I}}\right) \beta_{\,\textrm{I}}}{\delta(\delta+\gamma^{1,\textrm{I}}+\gamma^{2,\textrm{I}})}\,, 
  		\end{cases}
  		\label{eq11}
  		\vspace{-0.2em}
	\end{equation}
	we arrive at the following stability conditions:
	\begin{itemize}[leftmargin=*,labelindent=2.5mm,labelsep=3.0mm]
	\setlength\itemsep{0.2em}
		\item Case I: If $\tilde{R}_0^{(1)} \!>\! 1$, then the DFE is unstable for all $\gamma^1$.
		\item Case II: If $\tilde{R}_0^{(1)\!} \!+\! \tilde{R}_0^{(2)\!} \!<\! 1$, the DFE is stable for all~$\gamma^1$.
		\item Case III: If $\tilde{R}_0^{(1)} \!<\! 1$ and $\tilde{R}_0^{(1)} \!+\! \tilde{R}_0^{(2)} > 1$, then there exists some $\gamma^{1*} > 0$, where for $\gamma^1 < \gamma^{1*}$, the DFE is stable.
	\end{itemize}
\end{prop}
\begin{proof}
	See Appendix~\ref{app_D}. 
\end{proof}

The following section will focus on optimizing the activity rates of both susceptible and carrier nodes in a generalized~scenario, i.e., under heterogeneous activity patterns.\vspace{-0.3em}

\section{Optimal Activity Control}
\label{sec:control}
\fontdimen2\font=0.50ex
We now focus on controlling the activity cost of susceptible and carrier nodes while minimizing the spread of infection in the network.  In the context of activity homogeneity, studied in the preceding section, this involves minimizing the activity rate $\gamma^1$ to reduce the basic reproduction number within a given budget $(\mathcal{C})$. Hence, the homogeneous optimization problem is:\vspace{-0.4em}
\begin{subequations}
	\begin{empheq}[]{align}
		\mathcal{P}_{hom}:\,\,\, \min_{\gamma^1} \quad & \mathcal{R}_0=\rho(\mathbf{L})  \label{eq:20a}  \\
		\textrm{s.t.} 
		\quad & f(\gamma^1) \leq \mathcal{C},  \label{eq:20b} \\
  							& \underline{\gamma}^1 \leq \gamma^1 \leq \overline{\gamma}^1,\quad \forall i \in V,  \label{eq:20c}
  	\end{empheq}
  	\label{eq20:Homo}
\end{subequations}
\!\!where $f(\gamma^1)$ denotes the total activity cost and $\underline{\gamma}^1$ ($\overline{\gamma}^1$) is the minimum (maximum) feasible activity rate. Note that $\rho(\mathbf{L})$~is~a function of $\gamma^1$ (derived in \eqref{eq16n} and \eqref{eq:roLn}). Also, assuming that $f(\gamma^1)$ is monotonic, the constraint \eqref{eq:20b} lower-bounds~the~feasible values of $\gamma^1$, similar to $\gamma^1$ in \eqref{eq:20c}. The uni-variate~optimization problem of \eqref{eq20:Homo} can be readily solved using off-the-shelf optimization methods. 

The challenge, however, lies in optimizing the activity rates with heterogeneity. Before embracing this problem, it should~be noted that the spectral condition for convergence toward the DFE is specified by either $\mathcal{R}_0 \!=\! \rho(\mathbf{F}\mathbf{V}^{-1})$ or $\lambda_1(\mathbf{F} \!-\! \mathbf{V})$, where the latter is the largest real part of the eigenvalues of the matrix $\mathbf{Q} = \mathbf{F} -\mathbf{V}$ and indicates the exponential rate at which the infection dies out. That~is~to say, $\mathcal{R}_0 \!<\! 1$ if and only if $\lambda_1(\mathbf{Q}) \!<\! 0$. To~mitigate the infection spread in the heterogeneous setting, we thus focus on minimizing the value of $\lambda_1(\mathbf{Q})$ rather than $\mathcal{R}_0$, subject~to the constraint that the activity rates cost less than a predetermined threshold.

To determine $\mathbf{Q}$, we indicate that the DFE in the~heterogeneous case yields the steady-state probabilities $S_i^1 \!=\! \gamma_i^2/(\gamma_i^1 \!+\! \gamma_i^2)$ and $S_i^2 \!=\! \gamma_i^1/(\gamma_i^1 \!+\! \gamma_i^2)$, for any $i \!\in\! V$, where $\gamma_i^1$ and $\gamma_i^2$ are the activation and deactivation rates of node $i$ in the susceptible and carrier states, respectively, and the other state probabilities are equal~to zero. Also, we define vectors $\boldsymbol{\hat{\gamma}}^1 \triangleq (\gamma^1_i)_{i=1}^N$, $\boldsymbol{\hat{\gamma}}^2 \!\triangleq\! (\gamma^2_i)_{i=1}^N$, $\boldsymbol{\hat{\gamma}}^{1,\textrm{I}} \!\triangleq\! (\gamma^{1,\textrm{I}}_i)_{i=1}^N$, $\boldsymbol{\hat{\gamma}}^{2,\textrm{I}} \!\triangleq\! (\gamma^{2,\textrm{I}}_i)_{i=1}^N$, and $\mathbf{\hat{B}} \!\triangleq\! [p_{i,j}\, b_{i,j}]$. In point~of fact, $\mathbf{\hat{B}}$ is the weighted~adjacency matrix of the second layer, which incorporates $p_{i,j}$ as the weight of link $(i,j) \!\in\! E_2$.
\begin{lem}
	\label{lemma3}
	For the network with heterogeneous activity rates, $\mathbf{Q} \triangleq \mathbf{F} + \mathbf{V}^+ - \mathbf{V}^-$ such that:\vspace{-0.2em}
	\begin{align}
		\mathbf{F} &=
					\begin{bmatrix}
						\kappa\, \beta_{\textrm{C}}\, \mathbf{F}_1 & \kappa\, \beta_{\textrm{I}}\, \mathbf{F}_1\\
						\bar{\kappa}\, \beta_{\textrm{C}}\, \mathbf{F}_1 & \bar{\kappa}\, \beta_{\textrm{I}}\, \mathbf{F}_1
					\end{bmatrix}, \nonumber \\
		\mathbf{F}_1 &=
					\begin{bmatrix}
						\mathbf{A}_1 & \mathbf{A}_1\\
						\mathbf{A}_2 & \mathbf{A}_2 + \mathbf{B}_2
					\end{bmatrix}, \nonumber \\
		\mathbf{V}^+ &=
					\begin{bmatrix}
						\mathbf{0} & \mathbf{\Gamma}^2 & \mathbf{0} & \mathbf{0} \\
						\mathbf{\Gamma}^1 & \mathbf{0} & \mathbf{0} & \mathbf{0} \\
						\eta\, \mathbf{I} & \mathbf{0} & \mathbf{0} & \mathbf{\Gamma}^{2,{\textrm{I}}} \\
						\mathbf{0} & \eta\,\mathbf{I} & \mathbf{\Gamma}^{1,{\textrm{I}}} & \mathbf{0}
					\end{bmatrix}, \nonumber \\
		\mathbf{V}^- &= \Diag\left(\mathbf{\Gamma}^1 \!+\! \eta'\, \mathbf{I},\, \mathbf{\Gamma}^2 \!+\! \eta'\, \mathbf{I},\, \mathbf{\Gamma}^{1,{\textrm{I}}} \!+\! \delta\, \mathbf{I},\, \mathbf{\Gamma}^{2,{\textrm{I}}} \!+\! \delta\, \mathbf{I}\right), \nonumber
  \vspace{-0.2em}
	\end{align}
where the sub-matrices $\mathbf{A}_1 \!=\! \Diag\!\left(\boldsymbol{\hat{\gamma}}^2/(\boldsymbol{\hat{\gamma}}^1 \!+\! \boldsymbol{\hat{\gamma}}^2)\right)\,\mathbf{A}$, $\mathbf{A}_2 \!=\! \Diag\!\left(\boldsymbol{\hat{\gamma}}^1/(\boldsymbol{\hat{\gamma}}^1 \!+\! \boldsymbol{\hat{\gamma}}^2)\right)\,\mathbf{A}$, $\mathbf{B}_2 \!=\! \Diag\!\left(\boldsymbol{\hat{\gamma}}^1/(\boldsymbol{\hat{\gamma}}^1 \!+\! \boldsymbol{\hat{\gamma}}^2)\right)\,\mathbf{\hat{B}}$, $\mathbf{\Gamma}^1 \!=\! \Diag\left(\boldsymbol{\hat{\gamma}}^1\right)$, $\mathbf{\Gamma}^2 \!=\! \Diag\left(\boldsymbol{\hat{\gamma}}^2\right)$ , $\mathbf{\Gamma}^{1,{\textrm{I}}} \!=\! \Diag\left(\boldsymbol{\hat{\gamma}}^{1,{\textrm{I}}}\right)$, and $\mathbf{\Gamma}^{2,{\textrm{I}}} \!=\! \Diag\left(\boldsymbol{\hat{\gamma}}^{2,{\textrm{I}}}\right)$.
\end{lem}\vspace{-0.4em}
\begin{proof}
	See Appendix~\ref{app_E}.
\end{proof}

Following from Lemma~\ref{lemma3}, the budget-constrained optimization of the activity rates can be expressed as below:\vspace{-0.4em}
\vspace{0em}
\begin{subequations}
	\begin{empheq}[]{align}
		\mathcal{P}_{het_1}:\,\,\, \min_{\{\gamma_i^1\}} \quad & \lambda_1(\mathbf{Q})  \label{eq12a}\\
		\textrm{s.t.} \quad & \sum_{i=1}^N f_i(\gamma_i^1) < \mathcal{C},  \label{eq12b}\\
  							& \underline{\gamma}_i^1 \leq \gamma_i^1 \leq \overline{\gamma}_i^1,\quad \forall i \in V,  \label{eq12c}
  	\end{empheq}
  	\label{eq12}
\end{subequations}
\!\!where $\underline{\gamma}_i^1\, (\overline{\gamma}_i^1)$ denotes the minimum (maximum) activity rate of node $i$ in the susceptible and carrier states and~$f_i(\gamma_i^1)$~is the cost associated with node $i$, which is non-decreasing in the interval $[\underline{\gamma}_i^1, \overline{\gamma}_i^1]$. By applying the well-known Perron-Frobenius theorem, \eqref{eq12} can rightly be expressed as a \emph{geometric programming} (GP)\footnote{In a GP problem, both the objective function and inequality constraints must be posynomials, while equality constraints should be monomials. A function $h:\,\mathbb{R}^n_{++} \!\rightarrow\! \mathbb{R}$ is a \emph{monomial} if it~takes the form $h \!=\! c x_1^{a_1}x_2^{a_2}\ldots x_n^{a_n}$, where $c \!>\! 0$, and $\forall i \!\in\! V,\, a_i \!\in\! \mathbb{R}$. The sum of any monomials is a \emph{posynomial}.} problem \cite{Nowzari2017}, subject to two~prerequisites; firstly, $\mathbf{Q}$ must be a non-negative matrix and secondly, the entries of $\mathbf{Q}$ as well as $f_i(\cdot),~\forall i \!\in\! V$, must be posynomial functions of $\gamma^1_i$'s. However, Lemma~\ref{lemma3} demonstrates~that $\mathbf{Q}$ does not meet the non-negative matrix criteria due to the presence of negative diagonal entries. Additionally, the terms $\gamma_i^1/(\gamma_i^1 \!+\! \gamma_i^2)$ and $\gamma_i^2/(\gamma_i^1 \!+\! \gamma_i^2)$ in $\mathbf{F}_1$ render the entries of $\mathbf{Q}$ as non-posynomial. Next, we introduce alterations to $\mathcal{P}_1$ to enable the exploitation of the~GP framework.

We define matrix $\mathbf{\hat{Q}} \!\triangleq\! \mathbf{Q} + \psi\, \mathbf{I}$, where $\psi$ is a fixed known parameter specifying the maximum over all entries of $\mathbf{V}^-$ defined in Lemma \ref{lemma3}. Hence, we have:\vspace{-0.3em}
\begin{equation}
    \psi = \max\bigl\{\overline{\gamma}_i^1 + \eta', \gamma_i^2 + \eta', \gamma_i^{1,\textrm{I}} + \delta, \gamma_i^{2,\textrm{I}} + \delta \bigr\}_{i=1}^N.
    \vspace{-0.2em}
\end{equation}

Using Lemma~\ref{lemma3}, matrix $\mathbf{Q}$ can now be written as $\mathbf{\hat{Q}} \!=\! \mathbf{F} \!+\! \mathbf{V}^+ \!+ \mathbf{D}$, where $\mathbf{D} \!=\! \psi\, \mathbf{I} -\! \mathbf{V}^-$. The matrix $\mathbf{\hat{Q}}$ satisfies the non-negative matrix condition as both $\mathbf{F} \!+\! \mathbf{V}^+$ and $\mathbf{D}$ are non-negative.~Furthermore, we have $\lambda_1(\mathbf{\hat{Q}}) \!=\! \lambda_1(\mathbf{Q}) \!+\! \psi$, and thus, we can minimize $\lambda_1(\mathbf{\hat{Q}})$ instead of $\lambda_1(\mathbf{Q})$. In this regard, we substitute matrix $\mathbf{Q}$ with $\mathbf{\hat{Q}}$ in \eqref{eq12}. Since $\mathbf{\hat{Q}}$ is non-negative, we apply the Perron-Frobenius theorem\footnote{As per the Perron-Frobenius theorem, since $\hat{\mathcal{Q}}$ is non-negative, it has a real and positive eigenvalue $\lambda$ that is equal to $\lambda_1(\hat{\mathbf{Q}})$ and satisfies the condition $\lambda = \text{inf}\{\lambda'| \hat{\mathbf{Q}} u \preceq \lambda' u ~\text{for some } u  \succ 0 \}$. Therefore, instead of minimizing $\lambda_1(\hat{\mathbf{Q}} )$, $\lambda$ is minimized given that it satisfies $\hat{\mathbf{Q}} u \preceq \lambda u$ for some positive $u$, i.e., the first constraint of $\mathcal{P}_2$.} to express $\mathcal{P}_{het_1}$ as:
\vspace{-0.6em}
\begin{subequations}
	\begin{empheq}[]{align}
	\!\mathcal{P}_{het_2}:\,\, \min_{\{\gamma_i^1, \lambda, u_i\}} \quad & \lambda  \label{eq13a}  \\
		\textrm{s.t.} \qquad & \frac{\sum_{j=1}^N \mathbf{\hat{Q}}_{i,j}\left(\gamma_i^1\right) u_j}{\lambda\, u_i} \leq 1,\quad \forall i \in V,  \label{eq13b}  \\
		\quad & \sum_{i=1}^N f_i(\gamma_i^1) \leq \mathcal{C},  \label{eq13c} \\
  							& \underline{\gamma}_i^1 \leq \gamma_i^1 \leq \overline{\gamma}_i^1,\quad \forall i \in V.  \label{eq13d}
  	\end{empheq}
  	\label{eq13}
  	\vspace{-0.4em}
\end{subequations}

The second prerequisite for the GP framework, with posynomial entries in $\hat{\mathbf{Q}}$, is currently not fulfilled due to the presence of non-posynomial terms $\gamma_i^1/(\gamma_i^1 \!+\! \gamma_i^2)$ and $\gamma_i^2/(\gamma_i^1 \!+\! \gamma_i^2)$. Also, transforming matrix $\mathbf{Q}$ in \eqref{eq12} into a non-negative matrix leads to the emergence of new non-posynomial terms, namely $\psi-\eta'-\gamma^1_i$ for all $i$, which constitute the first $N$ diagonal entries of $\mathbf{D}$. These terms are non-posynomial due to the presence of variables $\gamma_i^1$ with negative coefficients. To ensure that all variable coefficients in the GP problem are positive, we introduce new variables $\zeta^1_i = \psi-\eta'-\gamma^1_i$ and replace the terms $\psi-\eta'-\gamma^1_i$ with~monomials $\zeta^1_i$. However, this introduces an equality~constraint $\gamma^1_i+\zeta^1_i=\psi-\eta'$ into the optimization problem,~which violates the condition that an equality constraint in~a GP problem should be in the form of a monomial equal to~a constant. By introducing~new variables $\zeta^1_i$, the modified formulation of $\mathcal{P}_{het_2}$ is given as follows:\vspace{-0.5em}
\begin{subequations}
	\begin{empheq}[]{align}
		\!\!\!\mathcal{P}_{het_3}\!:\,\, \min_{\{\gamma_i^1, \zeta^1_i, \lambda, u_i\}} \quad & \lambda  \label{eq14a}  \\
		\textrm{s.t.} \qquad\,\, & \frac{\sum_{j=1}^N \mathbf{\hat{Q}}_{i,j}\left(\gamma_i^1\right) u_j}{\lambda\, u_i} \leq 1,\quad \forall i \!\in\! V,\!  \label{eq14b}  \\
		\quad & \sum_{i=1}^N f_i(\gamma_i^1) \leq \mathcal{C},  \label{eq14c} \\
		\quad & \gamma_i^1 + \zeta_i^1 = \psi - \eta',\quad \forall i \!\in\! V,  \label{eq14d} \\
  		& \underline{\gamma}_i^1 \leq \gamma_i^1 \leq \overline{\gamma}_i^1,\quad \forall i \in V.  \label{eq14e}
  	\end{empheq}
  	\label{eq14}
  	\vspace{-0.3em}
\end{subequations}

The problem $\mathcal{P}_{het_3}$ cannot be directly solved using GP optimization because it contains non-posynomial terms in the entries of matrix $\hat{\mathbf{Q}}$ and a posynomial equality constraint in \eqref{eq14d}. However, it is possible to find a local solution using the iterative successive geometric programming (SGP) approach \cite{Boyd2007}. In each iteration of SGP, the posynomials that impede the problem from admitting the GP form are approximated by a monomial near the optimal solution of the previous iteration. This conversion of the optimization problem into~a GP enables it to be solved. Additionally, to ensure that the new optimal point does not deviate from the optimal point of the previous iteration, a trust region constraint is introduced. This is necessary because the monomial approximations are valid only in the vicinity of the previous optimal point. To approximate posynomials with monomials, we employ the lemma outlined below.\vspace{-0.4em}
\begin{lem}
	\label{lemma4}
	Given a posynomial $g(x) \!=\! \sum_{i}{u_i(x)}$ where~$\forall i$, $u_i(x)$ are monomials, we have $g(x) \!\geq\! \tilde{g}(x)=\Pi_i (u_i(x)/\alpha_i)^{\alpha_i}$. Also, if $\alpha_i \!=\! u_i(x_0)/g(x_0)$, then $g(x_0) \!=\! \tilde{g}(x_0)$, and $\tilde{g}(x)$ is the best monomial approximation of $g(x)$ near $x_0$ in the sense of the Taylor approximation \cite{Chiang2007}.
\end{lem}

Building on Lemma~\ref{lemma4}, we can express the optimization~problem in the $k$-th iteration as $\mathcal{P}^{(k)}_{het}$ in the following manner:\vspace{-0.4em}
\begin{subequations}
	\begin{empheq}[]{align}
		\mathcal{P}^{(k)}_{het}: \min_{\{\gamma_i^1, \zeta^1_i, \lambda, u_i\}} \,\,\, & \lambda  \label{eq15a}  \\
		\textrm{s.t.} \qquad & \frac{\sum_{j=1}^N \!\mathbf{\hat{Q}}_{i,j}^{(k)}\left(\gamma_i^1\right) u_j}{\lambda\, u_i} \leq 1,\quad \forall i \in V,  \label{eq15b}  \\
		\quad & \sum_{i=1}^N f_i(\gamma_i^1) \leq \mathcal{C},  \label{eq15c} \\
		\quad & \left(\gamma_i^1\right)^{\alpha_3}\! \left(\zeta_i^1\right)^{\alpha_4} \!=\! \psi - \eta',\quad \forall i \in V,  \label{eq15d} \\
  		& \underline{\gamma}_i^1 \leq \gamma_i^1 \leq \overline{\gamma}_i^1,\quad \forall i \in V,  \label{eq15e}  \\
  		& 1/1.1 \gamma_i^{1,(k-1)} \leq \gamma_i^1 \leq 1.1 \gamma_i^{1,(k-1)},  \label{eq15f}  \\
  		& 1/1.1 \gamma_i^{1,(k-1)} \leq \zeta_i^1 \leq 1.1 \gamma_i^{1,(k-1)},  \label{eq15g}
  	\end{empheq}
  	\label{eq15}
\end{subequations}
\!\!where $\mathbf{\hat{Q}}^{(k)}$ is the same as $\mathbf{\hat{Q}}$ except for the term $\gamma^1_i \!+\! \gamma^2_i$ that is approximated by a monomial using Lemma~\ref{lemma4}. In particular, $g(x) \!=\! \gamma_i^1 \!+\! \gamma_i^2$ is approximated with:\vspace{-0.3em}
\begin{equation}
	\tilde{g}(x) = \left(u_1(x)/\alpha_1\right)^{\alpha_1} \left(u_2(x)/\alpha_2\right)^{\alpha_2},
	\label{eq:gapprox}
 \vspace{-0.3em}
\end{equation}
where $u_1(x) \!=\! x \!=\! \gamma_i^1$ and $u_2(x) \!=\! \gamma_i^2$. Moreover, $\alpha_1$ and~$\alpha_2$ are derived as follows, where $x_0 \!=\! \gamma_i^{1, (k-1)}$ denotes the optimal value of $\gamma^1_i$ in the $(k \!-\!1)$-th iteration (i.e., the solution of $\mathcal{P}^{(k-1)}$):\vspace{-0.3em}
\begin{equation}
  		\begin{cases}
    		\alpha_1 =& \cfrac{u_1(x_0)}{g(x_0)} =  \cfrac{\gamma_i^{1,(k-1)\!}}{\gamma_i^{1,(k-1)\!} \!+\! \gamma_i^2}\,,  \vspace{0.1em} \\
    		\alpha_2 =& \cfrac{u_2(x_0)}{g(x_0)} =  \cfrac{\gamma_i^{2\!}}{\gamma_i^{1,(k-1)\!} \!+\! \gamma_i^2}\,. 
  		\end{cases}
  		\label{eq26}
  		\vspace{-0.3em}
\end{equation}
By substituting \eqref{eq26} in \eqref{eq:gapprox}, $\gamma^1_i \!+\! \gamma^2_i$ can be approximated as:\vspace{-0.3em}
\begin{align}
	\gamma^1_i \!+\! \gamma^2_i &\simeq (\gamma_i^1/\gamma_i^{1,(k-1)})^{\alpha_1},
\label{eq27}
\vspace{-0.4em}
\end{align}
where $\alpha_1$ is a constant given in \eqref{eq26}. In a similar manner, the term $\gamma^1_i \!+\! \zeta^1_i$ in constraint \eqref{eq14d} can be approximated with $(\gamma^1_i/\alpha_3)^{\alpha_3} (\zeta^1_i/\alpha_4)^{\alpha_4}$. Here, $\alpha_3$ and $\alpha_4$ are obtained to be:\vspace{-0.3em}
\begin{equation}
  	\begin{cases}
    	\alpha_3 =& \!\cfrac{\gamma_i^{1,(k-1)\!}}{\gamma_i^{1,(k-1)} \!+ \zeta_i^{1,(k-1)}}\,,  \vspace{0.2em} \\
    	\alpha_4 =& \!\cfrac{\zeta_i^{1,(k-1)\!}}{\gamma_i^{1,(k-1)\!} \!+\! \zeta_i^{1,(k-1)}}\,. 
  	\end{cases}
  	\label{eq28}
  	\vspace{-0.3em}
\end{equation}
where $\zeta^{1,(k-1)}_i$ is the optimal value of $\zeta^{1}_i$ in the $(k \!-\! 1)$-th iteration. Finally, \eqref{eq15f} and \eqref{eq15g} are trust region constraints ensuring that the optimal point in iteration $k$ does not deviate from the preceding optimal point. The proposed SGP method is adaptable to diverse network structures, as it relies solely~on the knowledge of the network graph. While SGP may pose~challenges when directly applied to very large networks due to increased complexity, a viable strategy involves partitioning the large network into its connected components and subsequently applying SGP to each component.\vspace{-0.4em}

\section{Numerical Results and Discussion}
\label{sec:numerical}
\fontdimen2\font=0.50ex
We assess the impact of different key parameters on infection virality and determine the optimal activity rates of susceptible and asymptomatic nodes using MF numerical and simulation results. For epidemic outbreak analysis, we consider activity-homogeneous nodes while for the control mechanism, we evaluate a network~with activity heterogeneity. Simulations are conducted in MATLAB using the Gillespie algorithm \cite{Gillespie1976} adapted in \cite{Aram2020}. The base parametric settings, unless stated~otherwise, are provided in Table~\ref{table1}. Each simulation is averaged over $1000$ runs of~the epidemic process. The SGP is solved using the MOSEK~solver within the CVX package \cite{cvx}.\vspace{-0.8em}

\subsection{Infection Dynamics}
\label{sec:dynamics}
\fontdimen2\font=0.50ex
We synthesize two distinct random regular networks of $N \!=\! 500$ nodes for the $G_1$ and $G_2$ layers with degrees $d_1 \!=\!4$ and $d_2 \!=\! 50$, respectively. In~$G_2$, any link between two active nodes is activated with probability $p \!=\! 0.3$.
\begin{table}[!t]
	\centering
	\caption{Parametric values used for simulation.}
	\vspace{-0.4em}
	\renewcommand{\arraystretch}{1.2}
	\begin{tabular}{|c|c||c|c|}
		\hline
		Parameter & Value & Parameter& Value \\ 
		\hline \hline
		$N$   & $500$ & $(d_1, d_2)$ & $(4, 50)$  \\ 
		$\eta'$  & $0.8$  & $\eta$   & $0.56$  \\ 
        $(\beta_{\,\textrm{I}},\beta_{\,\textrm{C}})$   & $(0.2, 0.1)$ & $\delta$  & $1.5$   \\ 
       $\gamma^{1,\textrm{I}}$   & $0$   & $p$  & $0.3$  \\ 
       \hline
	\end{tabular}
	\label{table1}
 \vspace{-0.2em}
\end{table}
\begin{figure}[!t]
	\centering
	\includegraphics[width=0.9\columnwidth]{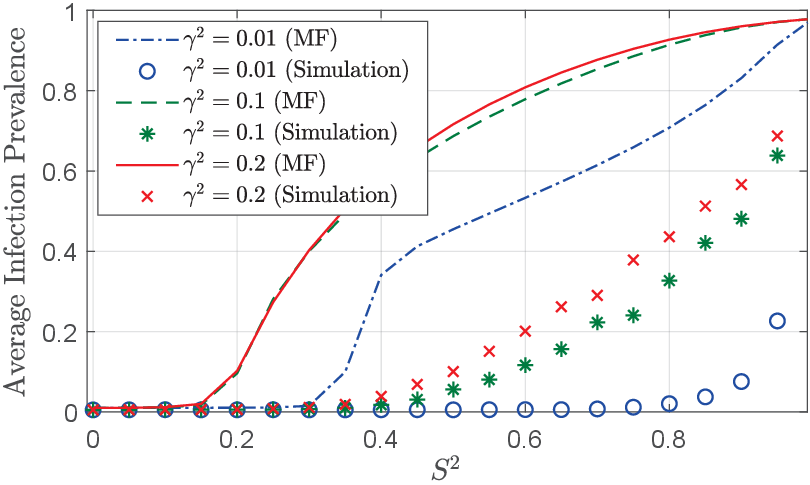}
	\vspace{-1em}
	\caption{Average infection prevalence versus the activity probability of~susceptible and carrier nodes ($S^2$) in steady-state for varying $\gamma^2$ values.}
	\label{fig3n}
    \vspace{-0.3em}
\end{figure}

\figurename{~\ref{fig3n}} shows the average infection prevalence, defined as~the steady-state probability of node recovery, in terms of the activity probability of susceptible and asymptomatic nodes ($S^2$) for~the original SCIR model (i.e., $\kappa \!=\! 1$). The results indicate that~the exact simulations are consistently upper-bounded by our MF approximation, which is in line with existing literature \cite{Aram2020}. Thus, the MF approximation invariably underestimates the actual network epidemic threshold. When the inactivity rate of susceptible and carrier nodes (${\gamma^2}$) remains constant, increasing $S^2$ generates more connections in $G_2$ on average as nodes spend more time in the active state, thus facilitating wider infection transmission in the second layer. Conversely, for a fixed $S^2$, as ${\gamma^2}$ increases, the rate $\gamma^1$ proportionally increases as well since $S^2 \!=\! \gamma^1/(\gamma^1 \!+\! \gamma^2)$. This implies that the nodes change their~activity states more~frequently, while their activity probability remains unchanged. Consequently, more active links are introduced in the second layer, leading to a wider infection spread. However, the rise in infection prevalence is less pronounced when $\gamma^2$ is higher because the activity rate $\gamma^1 \!=\! \gamma^2S^2/(1 \!-\! S^2)$ is also sufficiently high (with $S^2$ remaining constant), meaning that almost all possible connections in the network~become active.

\figurename{~\ref{fig4n}} plots the time evolution of the average population size in each epidemic state for $\eta'\!=\! 0.1$ and $\eta' \!=\! 0.4$,~respectively. The fixed ratio of $\eta/\eta' \!=\! 0.7$ is assumed,~where $\mathbb{E}\!\left[N_{\textrm{X}}\right]$ represents the mean number of nodes in the epidemic state $\textrm{X} \!\in\! \{\textrm{S}, \textrm{C}, \textrm{I}, \textrm{R}\}$. As apparent in \figurename{~\ref{fig4n}}(a), we observe that wider transmission occurs when $\eta'$ rates are lower because nodes tend to spend more time on average in the carrier state. This allows for a greater progression of the infection throughout the network. On the other hand, when $\eta'$ grows to $0.4$ in \figurename{~\ref{fig4n}}(b), the transition of nodes from carriers to infected transpires at a faster rate of $\eta \!=\! 0.28$. As a result, such nodes tend to curtail their activities, leading to a decrease in the infection spread.
\begin{figure}[!t]
	\centering
	\includegraphics[width=0.9\columnwidth]{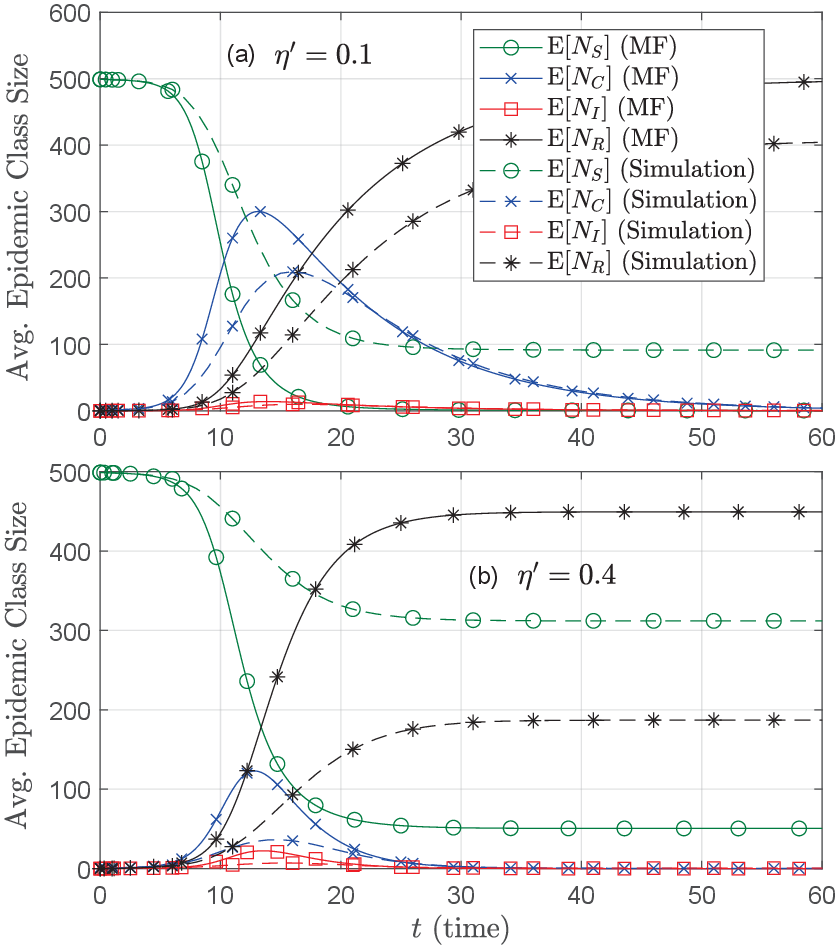}
	\vspace{-1em}
	\caption{Evolution of the average size of all epidemic compartments over time for varying rates at which carriers either get infected or recover ($\eta'$).}
	\label{fig4n}
    \vspace{-0.4em}
\end{figure}
\begin{figure}[t]
	\centering
	\includegraphics[width=0.9\columnwidth]{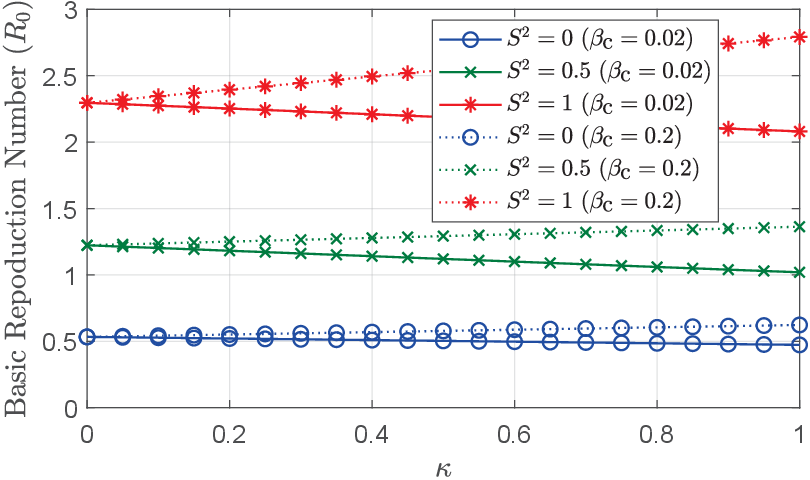}
	\vspace{-1em}
	\caption{Basic reproduction ratio ($\mathcal{R}_0$) versus $\kappa$ for varying $S^2$ and $\beta_{\textrm{C}}$ values.}
	\label{fig5n}
    \vspace{-0.3em}
\end{figure} 

\figurename{~\ref{fig5n}} portrays the relationship between the basic reproduction number ($\mathcal{R}_0$) and $\kappa$ in the extended SCIR model, with varying~$S^2$ and $\beta_{\textrm{C}}$ values. The figure shows that when the transmission rate $\beta_{\textrm{C}}$ is sufficiently low (e.g., $\beta_{\,\textrm{C}} \!=\! 0.02$), increasing $\kappa$ leads to~a decrease in $\mathcal{R}_0$. This decrease results from more susceptible nodes becoming carriers rather than infected, leading to a lower disease transmission probability compared to infected nodes. Conversely, for higher $\beta_{\textrm{C}}$ (around $0.2$), a larger proportion of carrier nodes contribute to higher disease prevalence since they exhibit the same activity level as susceptible nodes. Moreover, \figurename{~\ref{fig5n}} also reveals that the $\mathcal{R}_0$ threshold increases with $S^2$ as~more connections are formed in the $G_2$ layer.

To numerically verify Propositions~\ref{prop1} and \ref{prop2}, \figurename{~\ref{fig7n}} plots the average infection prevalence of the original and extended SCIR models as a function of $S^2$ for different values of $(d_1,d_2)$. Each value of $(d_1,d_2)$ corresponds to specific values of $R^{(1)}_0$ and $R^{(1)\!}_0+R^{(2)}_0$ ($\tilde{R}^{(1)}_0$ and $\tilde{R}^{(1)\!}_0+\tilde{R}^{(2)}_0$). We observe that when $R^{(1)}_0 \!+\! R^{(2)}_0 < 1$, as in the case of $(d_1,d_2) \!=\! (3,6)$, the disease does not~spread in the network for any values of $S^2$. Conversely, when $R^{(1)}_0 \!>\! 1$, the disease prevails for all values of $S^2$. In the case of $(d_1,d_2) \!=\! (3,12)$, where $R^{(1)}_0 \!+\! R^{(2)}_0 \!>\! 1$ and $R^{(1)}_0 \!<\! 1$, a critical value~$S^{2*}$ exists, corresponding to a specific~$\gamma^{1*}$, beyond which the disease prevails for $S^2 \!>\! S^{2*}$. The same observations hold true in the extended case (i.e., $\kappa \!=\! 0.6$), thereby confirming the findings stated in the two propositions.\vspace{-0.6em}
\begin{figure}[!t]
	\centering
	\includegraphics[width=0.9\columnwidth]{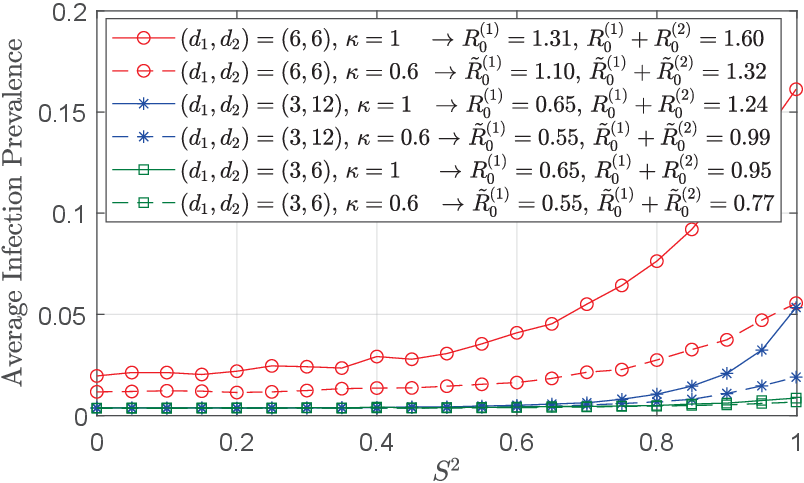}
	\vspace{-1em}
	\caption{Average infection prevalence versus the activity probability of carrier and susceptible nodes ($S^2$).}
	\label{fig7n}
 \vspace{-0.3em}
\end{figure}

\subsection{Activity Control Policy Assessment}
\label{sec:controlpolicy}
\fontdimen2\font=0.52ex
We now investigate the effect of limited cost on activity rates of the nodes and disease prevalence. We assume that $\gamma^2 \!=\! 0.2$,~$\beta_{\,\textrm{C}} \!=\! 0.15$, $\delta \!=\! 0.2$, and $\kappa \!=\! 1$. As discussed~in~Section~\ref{sec:dynamics}, $G_1$ remains a random regular graph. Moreover, in the second layer,~we assume $p_{i,j} \!=\! p_i\,p_j$, where $p_i$ is the probability that node $i$ decides to create a link with any of its direct active neighbors. For a connection between nodes $i$ and $j$ to become active in $G_2$, both nodes must independently decide to activate their connection with probabilities $p_i$ and $p_j$, respectively. Moreover, each node belongs to one of three classes, where the nodes in class $l' \!\in\! \{1,2,3\}$ form links with their active neighbors with a specific probability $p_{l'}$. We set $p_1 \!=\! 0.1$, $p_2 \!=\! 0.2$, and $p_3 \!=\! 0.8$ for demonstration. The number of nodes in each class, denoted as $N_{l'}$, is $N_1 \!=\! N_3 \!=\! \lfloor N/6 \rfloor$, and $N_2 \!=\! N \!-\! N_1 \!-\! N_3$. All nodes have a minimum activity rate of $\underline{\gamma}_i^1 \!=\! 0.08$ and a maximum activity rate of $\overline{\gamma}_i^1 \!=\! 0.3$.  The cost function related to node $i$ is defined as $1/\gamma^1_i$, resulting in a budget range of $\frac{N}{0.3} \!\leq \mathcal{C} \leq\! \frac{N}{0.08}$.

\figurename{~\ref{fig6nn}}(a) and \figurename{~\ref{fig6nn}}(b) depict the optimal activity rates of~nodes relative to their average~degrees, calculated as $\sum_{j\in N} a_{i,j}+p_{i,j} b_{i,j}$, across three distinct classes for network sizes $N=200$ and~$N=400$, respectively.~Here, layer $G_2$ of the synthesized network is modeled as a Barab\'{a}si-Albert (BA) graph with an~initial seed size of $20$ nodes and a preferential attachment~parameter of $10$, resulting in a heterogeneous degree distribution.~The initial $20$ seed nodes are specifically chosen from the~second class.~The available budget per node is almost equal to $6$,~leading to a total budget of $1216$ and $2432$ for $N \!=\! 200$ and $N \!=\! 400$, respectively. It is evident that these~$20$ nodes with the largest~average degrees are assigned the lowest activity rates in both cases of $N \!=\! 200$ and $N \!=\! 400$.~Furthermore, nodes in the third~class are also~assigned low activity rates due to their linkage to neighbors with~the highest probabilities (i.e., $0.8$), resulting in relatively higher average degrees. Nodes from the first class and non-initial seed nodes from the second class are assigned activity rates ranging from the highest possible value (i.e., $0.3$) to $0.11$, based on their average degrees.~In general, there is a decrease in activity rates as nodes exhibit higher average degrees.~In \figurename{~\ref{fig6nn}}(b), more number of third-class nodes are assigned the lowest activity rates for $N \!=\! 400$, as the average~degrees of these nodes also increase with the network size.
\begin{figure}[!t]
	\centering
	\includegraphics[width=1.0\columnwidth]{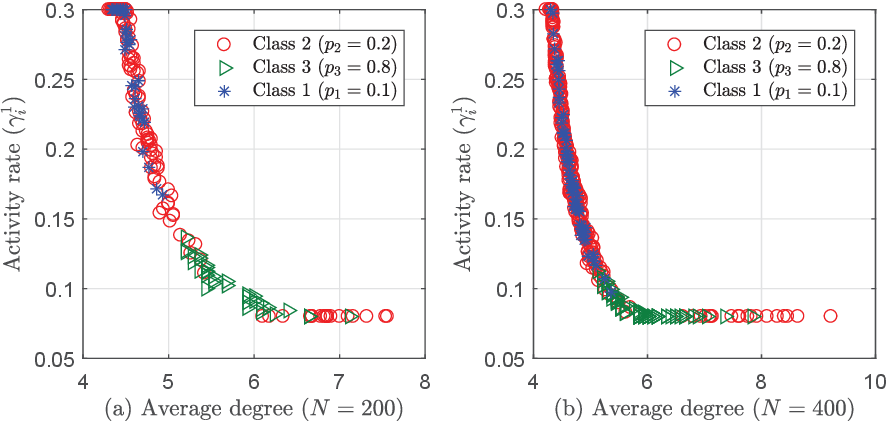}
	\vspace{-1.9em}
	\caption{Optimal activity rate ($\gamma^1_i$) versus the average node degree.}
	\label{fig6nn}
\end{figure}

\begin{figure}[!t]
	\centering
	\includegraphics[width=0.938\columnwidth]{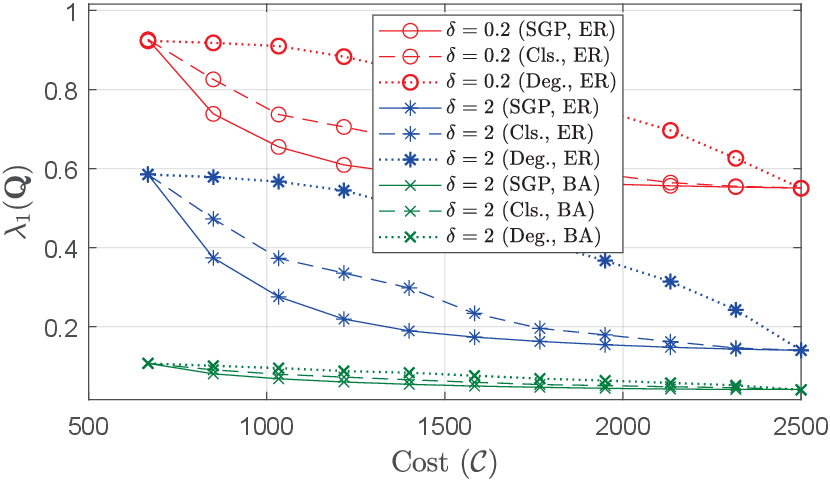}
	\vspace{-1em}
	\caption{Spectral radius ($\lambda_1(\mathbf{Q})$) versus budget ($\mathcal{C}$) under varying $\delta$ values.}
	\label{fig8n}
    \vspace{-0.6em}
\end{figure}

\figurename{~\ref{fig8n}} shows the optimal spectral abscissa $\lambda_1(\mathbf{Q})$ in~terms of the available budget ($\mathcal{C}$) for $\gamma^{1,\textrm{I}} \!=\! 0.2$ and under three~allocation policies: SGP, degree centrality-based (Deg.), and closeness centrality-based (Cls.) allocations. Layer $G_2$ of the synthesized network is modeled as an Erd\"{o}s-Renyi (ER) graph with~a~connection probability of $0.2$ as well as the BA network of \figurename{~\ref{fig6nn}}. In Deg., nodes with higher average degrees receive budget~allocation first, while~Cls. prioritizes nodes with higher closeness measure values, indicating their proximity to other nodes in the network. 
The figure shows that, for both low and high budget values,~all policies exhibit similar performance, assigning the highest and lowest activity rates to nodes, respectively. However, for intermediate cost values, SGP outperforms~Cls. by up~to $13\%$ and $36\%$ for $\delta \!=\! 0.2$ and $\delta \!=\! 2$, respectively.~The corresponding improvements over Deg. are nearly $30\%$ and $62\%$. This is because SGP considers both the temporal and structural properties of the network by utilizing~the adjacency matrix and activity rates in \eqref{eq15}, unlike the centrality methods. The higher improvement at $\delta \!=\! 2$ is attributed to the significantly greater cure rate of infected nodes compared to the activity rates, resulting in lower epidemic spread for the same activity rates at the same cost. Another observation to be made from \figurename{~\ref{fig8n}} is that the second layer of the network is modeled using a BA graph instead of an ER graph. The BA graph has an initial seed size of $20$, an attachment parameter of $10$, and $\delta \!=\! 2$. Comparing SGP to Cls. and Deg. policies,~SGP shows improvements of up to $17\%$ and $35\%$, respectively. Furthermore, with $\delta=2$, the observed improvement in the ER network compared to the BA~network is linked to their density disparities under the given parameters.~The higher density of the ER network amplifies the influence of temporal links in disease propagation.~Hence, by accounting for node activity patterns to optimize activity rates, SPG demonstrates superior performance compared to other methods.\vspace{-0.6em}

\subsection{Evaluation using Real-World Dataset}
\label{sec:datset}
\fontdimen2\font=0.50ex
We validate our findings using the Ebola virus disease (EVD) dataset \cite{Riad2019}, which covers $23$ districts in Uganda and comprises $11{,}056$ nodes. It consists of a permanent layer that represents connections within households and a temporal layer incorporating movement patterns of individuals entering district borders and heading towards the capital city. Each moving node is considered active. Following \cite{Riad2019}, we assume link construction probabilities of~$p_{i,j} \!=\! 0.7$, and set $\gamma^{1,\textrm{I}} \!=\! 0$ as in the previous section. We selected a connected component of the first layer with $572$ nodes and $2{,}848$ links for simplicity. Consequently, the connections of the corresponding nodes in the second layer form $G_2$, consisting of $828$ links. For disease modeling, we use parameters estimated for COVID-19 in \cite{Prem2020}, i.e., we set the probabilities of a carrier node becoming infectious or recovered ($\eta'$) and of an infected node recovering ($\delta$) to be $0.14$ and $0.18$, respectively. Following \cite{Kuzdeuov2020}, we also set $\beta_{\textrm{C}} \!=\! 0.7 \beta_{\textrm{I}}$ and $\eta \!=\! 0.7 \eta'$.

\figurename{~\ref{fig9n}} depicts the disease prevalence against $S^2$ in the EVD dataset. For $\beta_{\textrm{I}} \!=\! 0.01$, the disease dies out for all~values of $S^2$. Interestingly, for $\beta_{\textrm{I}} \!=\! 0.03$, the disease spread increases with $S^2$ similar to the observation in \figurename{~\ref{fig3n}}. Additionally, the prevalence generally increases with $\gamma^2$ for any fixed value of $S^2$ due to more frequent activation and deactivation of nodes, thus leading to more activated links in the second layer. However, for high values of $S^2$, the prevalence decreases with $\gamma^2$ instead. In such cases, the shorter duration of active links compared to the disease transmission time ($1/\beta_{\textrm{I}}$) leads to frequent switching between activity states, reducing the likelihood of disease transmission.\vspace{-0.6em}
\begin{figure}[!t]
	\centering
	\includegraphics[width=0.9\columnwidth]{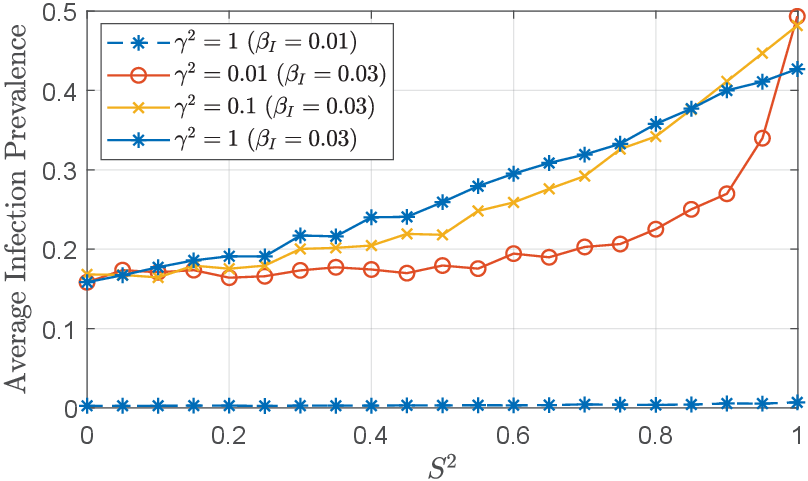}
	\vspace{-1em}
	\caption{Average infection prevalence versus the activity probability of~susceptible and carrier nodes ($S^2$) for varying $\gamma^2$ values in the EVD dataset.}
	\label{fig9n}
    \vspace{-0.3em}
\end{figure}

\section{Conclusion}
\label{sec:conclusion}
\fontdimen2\font=0.50ex
In this paper, the impact of asymptomatic carriers on the~epidemic threshold for a stochastic SCIR model over a two-layer temporal social network was investigated. The networked model comprised distinct dynamic and static connection layers, with probabilistic interactions between active nodes in the former layer. Sufficient conditions were derived for the infection to~die out or persist in a network with activity pattern homogeneity. Furthermore, an optimal activity control mechanism was formulated to contain the activity rates of susceptible and asymptomatic nodes, with their activity profiles being indistinguishable. Their activity rates were optimized by minimizing the spectral radius of the proposed mean-field approximated network model under limited budgetary constraints. Results from simulation experiments showed that the proposed SGP approximation outperformed the degree and closeness centrality-based benchmarks. Promising future research directions include modeling the effect of asymptomatic carriers on disease mutation and devising learning-based control policies that are both scalable and effective.\vspace{-0.5em}

\appendices
\section{Proof of Lemma~\ref{lemma1}}
\label{app_A}
\fontdimen2\font=0.50ex
To compute $\mathcal{R}_0$ for \eqref{eq:MFhomo}, the matrices $\mathbf{F}$ and $\mathbf{V}$ are given as:\vspace{-0.2em}
\begin{align}
	\mathbf{F} \!=\!
			\begin{bmatrix}
				\beta_{\textrm{C}}\, \mathbf{F}_{\!1} & \!\!\beta_{\textrm{I}}\, \mathbf{F}_{\!1} \\
				\mathbf{0} & \!\!\mathbf{0}
			\end{bmatrix}\! \nonumber \text{\,and\,} 
	\mathbf{V} \!=\!
			\begin{bmatrix}
				\gamma^{1} \!+\! \eta' & \!\!\!-\gamma^{2} & \!\!\!0 & \!\!\!0 \\
				-\gamma^{1} & \!\!\!\gamma^{2} \!+\! \eta' & \!\!\!0 & \!\!\!0 \\
				-\eta & \!\!\!0 & \!\!\!\gamma^{1,\textrm{I}} \!+\! \delta \!\!\!& -\gamma^{2,\textrm{I}\!} \\
				0 & \!\!\!-\eta & \!\!\!-\gamma^{1,\textrm{I}\!} & \!\!\!\gamma^{2,\textrm{I}} \!+\! \delta
			\end{bmatrix}, \nonumber
   \vspace{-0.2em}
\end{align}
where $\mathbf{F}_{\!1} \!=\!
			\begin{bmatrix}
				d_1\, S^1 & \!\!d_1\, S^1\\
				d_1\, S^2 & \!\!(p\,d_2 \!+\! d_1)\, S^2
			\end{bmatrix}$. Rewriting $\mathbf{V}$ in terms of $\mathbf{V}_{\!1}$ and $\mathbf{V}_{\!2}$ yields $\mathbf{V} \!=\!
					\begin{bmatrix}
						\mathbf{V}_{\!1} & \!\mathbf{0}  \\
						-\eta\, \mathbf{I} & \!\mathbf{V}_{\!2}
					\end{bmatrix}$,
where $\mathbf{V}_{\!1} \!=\!
					\begin{bmatrix}
						\gamma^{1\!} \!+\! \eta'\!\! & \!\!-\gamma^2  \\
						-\gamma^{1}\!\! & \!\!\gamma^{2\!} \!+\! \eta'
					\end{bmatrix}$
					and 
					$\mathbf{V}_{\!2} \!=\!
					\begin{bmatrix}
						\gamma^{1,\textrm{I}\!} \!+\! \delta \!&\! -\gamma^{2,\textrm{I}}  \\
						-\gamma^{1,\textrm{I}} \!&\! \gamma^{2,\textrm{I}\!} \!+\! \delta
					\end{bmatrix}$. 
Taking the inverse of $\mathbf{V}$, we get $\mathbf{V}^{-1} \!=\!
					\begin{bmatrix}
						\mathbf{V}_{\!1}^{-1} & \!\!\mathbf{0}  \\
						\eta\,\mathbf{V}_{\!2}^{-1}\mathbf{V}_{\!1}^{-1} & \!\!\mathbf{V}_{\!2}^{-1}
					\end{bmatrix}$.
Hence, the product $\mathbf{F \cdot V}^{-1}$ is:\vspace{-0.2em}
\begin{align}
		\mathbf{F\cdot V}^{-1} \!=\!
					\begin{bmatrix}
						\mathbf{F}_1\cdot \left(\beta_C\, \mathbf{I} \!+\! \eta\,\beta\, \mathbf{V}_2^{-1}\right)\cdot \mathbf{V}_1^{-1} & \beta\,\mathbf{F}_1\cdot \mathbf{V}_2^{-1}  \\
						\mathbf{0} & \mathbf{0}
					\end{bmatrix}. \nonumber
     \vspace{-0.2em}
\end{align}
To obtain $\rho(\mathbf{F \!\cdot\! V}^{-1})$, we need to calculate the spectral~radius or, equivalently, the eigenvalues of $\mathbf{L} \!=\! \mathbf{F}_1 \cdot (\beta_{\textrm{C}}, \mathbf{I} + \eta\,\beta\,\mathbf{V}_{\!2}^{-1}) \cdot \mathbf{V}_{\!1}^{-1}$.\vspace{-0.4em}

\section{Proof of Proposition~\ref{prop1}}
\label{app_B}
\fontdimen2\font=0.50ex
Lemma~\ref{lemma1} introduces the spectral radius of matrix $\mathbf{L}$, defined as the largest absolute value of its eigenvalues, which can be found by solving the equation $\det\left(\mathbf{L} - \lambda \mathbf{I}\right) \!=\! \lambda^2 - \Tr\left(\mathbf{L}\right) + \det\left(\mathbf{L}\right)$. To obtain these eigenvalues, we first need to determine the entries of matrix $\mathbf{L}$, which after some mathematical manipulations can be expressed as:
\begin{equation}
	\begin{cases}
		\mathbf{L}_{1,1} \!=\! \mathbf{L}_{1,2} \!=\! S^1 R_0^{(1)}, \vspace{0.2em} \\
		\mathbf{L}_{2,1} \!=\! S^2 \Big(\!R_0^{(1)} \!+\! R_0^{(2)} \!-\! \cfrac{p\, d_2\, \eta\, (\eta' \!+\! \gamma^{2})}{\eta' (\eta' \!+\! \gamma^{1} \!+\! \gamma^{2})}\,w \Big),  \vspace{0.2em} \\
		\mathbf{L}_{2,2} \!=\! S^2 \Big(\!R_0^{(1)} \!+\! R_0^{(2)} \!-\! \cfrac{p\, d_2\, \eta\, \gamma^{2}}{\eta' (\eta' \!+\! \gamma^{1} \!+\! \gamma^{2})} \,w \Big), 
	\end{cases}
	\label{eq10}
\end{equation}
where $w \!=\! \beta_{\textrm{C}} / \eta+\beta_{\textrm{I}}/( \delta+\gamma^{1,\textrm{I}}+\gamma^{2,\textrm{I}})$. In what follows, we use \eqref{eq10} to show that the eigenvalues of $\mathbf{L}$ are real. From \eqref{eq10}, we arrive at the following expressions for $\Tr\left(\mathbf{L}\right)$ and $\det\left(\mathbf{L}\right)$:
\vspace{0em}
\begin{equation}
	\!\!\begin{cases}
		\!\Tr(\mathbf{L}) \!=\! R_0^{(1)} \!+\! S^2 \overbrace{\Big(\!R_0^{(2)} \!-\! \cfrac{p\,d_2\,\eta\,\gamma^{2}}{\eta'\,(\eta' \!+\! \gamma^{1} \!+\! \gamma^{2})}\,w\Big)}^{z_1}, \vspace{0.2em}   \\
		\!\det(\mathbf{L}) = S^1\,S^2\,R_0^{(1)}\,\underbrace{\cfrac{p\,d_2\,\eta}{\eta' + \gamma^{1} + \gamma^{2}}\,w}_{z_2}.  
	\end{cases}
	\label{eq16n}
 \vspace{-0.2em}
\end{equation}
Rewriting $\Tr(\mathbf{L})$ and $\det(\mathbf{L})$ in terms of variables $z_1$ and $z_2$, respectively, reduces \eqref{eq16n} to:\vspace{-0.2em}
\begin{equation}
	\begin{cases}
		\Tr(\mathbf{L}) = R_0^{(1)} + S^2 z_1, \vspace{0.2em}   \\
		\det(\mathbf{L}) = S^1\,S^2\,R_0^{(1)}z_2.  
	\end{cases}
 \vspace{-0.2em}
\end{equation}
Subsequently, the discriminant of $\text{det}(\mathbf{L} \!-\! \lambda\mathbf{I})$, given as $\Delta \!\triangleq\! \Tr(\mathbf{L})^2 \!-\! 4\det(\mathbf{L})$, can be expressed as:\vspace{-0.2em}
\begin{equation}
	\Delta \!=\! (R_0^{(1)} \!+\! S^2 z_1 \!-\! 2 S^1S^2\,z_2)^2 + 4S^1(S^2)^2 z_2(z_1 \!-\! S^1 z_2).
 \vspace{-0.2em}
\end{equation}
For $\Delta \geq 0$, we need to prove that $z_1 \geq z_2$. To do so, we can recast the inequality to $w\,\gamma^{1} + \beta_{\textrm{I}}\,\gamma^{1,\textrm{I}}(\eta' + \gamma^{1} + \gamma^{2})/\big(\delta(\delta + \gamma^{1,\textrm{I}} + \gamma^{2,\textrm{I}})\big) \geq 0$ by substituting the values for $z_1$ and $z_2$, and using $R_0^{(2)}$ from \eqref{eq:R02}. Since $\Delta$ is always positive, we establish that $\rho(\mathbf{L})$ is equal to:\vspace{-0.2em}
\begin{equation}
	\big(\Tr(\mathbf{L}) + \sqrt{\Delta}\big)/2 = \left(\!\Tr(\mathbf{L}) + \sqrt{\Tr(\mathbf{L})^2 \!-\! 4\det(\mathbf{L})}\right)\!/2.
	\label{eq:roLn}
 \vspace{-0.2em}
\end{equation}

We now prove the three stability conditions in the proposition.
\begin{itemize}
	\item \textit{Case I}: Letting $\rho(\mathbf{L}) \!<\! 1$ yields the inequality $\Tr(\mathbf{L}) \!-\! \det(\mathbf{L}) \!<\! 1$, which can further be rewritten as:\vspace{-0.2em}
	\begin{equation}
		R_0^{(1)} \!+\! S^2 R_0^{(2)} \!-\! 1 < \cfrac{S^2\, p\, d_2\, \eta}{\eta' \!+\! \gamma^{1} \!+\! \gamma^{2}}\, w\Big(\cfrac{\gamma^2}{\eta'} + S^1\,R_0^{(1)\!}\Big).
		\label{eq17}
  \vspace{-0.2em}
	\end{equation}
It is important to observe that $S^1$ and $S^2$ in \eqref{eq17} depend on $\gamma^1$. Also, note that the right-hand side of the inequality is greater than zero for all possible values of $\gamma^1$. Consequently, if the condition $R_0^{(1)} + R_0^{(2)} < 1$ holds, then the left-hand side of the inequality will be negative. As a result, the inequality~will hold for any $\gamma^1 > 0$. This implies that the DFE is stable at all times.
	\item \textit{Case II}: To prove this scenario, it suffices to demonstrate that $\mathcal{R}_0 \!\geq\! R_0^{(1)}$ for any $\gamma^1$ value. Substituting $\mathcal{R}_0$ with $(\Tr(\mathbf{L}) \!+\! \sqrt{\det(\mathbf{L})^2 - 4 \det(\mathbf{L})},)/2$ allows us to express the inequality $\mathcal{R}_0 \!\geq\! R_0^{(1)}$ as follows:\vspace{-0.2em}
	\begin{equation}
		R_0^{(1)} \Tr(\mathbf{L}) - \det(\mathbf{L}) \geq \left(R_0^{(1)}\right)^2.
		\label{eq18}
  \vspace{-0.2em}
	\end{equation}
	Building on the fact that $\text{det}(\mathbf{L}-\lambda\mathbf{I}) > 0$, \eqref{eq18} can be simplified into:\vspace{-0.2em}
	\begin{equation}
		R_0^{(2)} - \cfrac{p\,d_2\,\eta\,\gamma^2}{\eta'(\eta' \!+\! \gamma^{1} \!+\! \gamma^{2})}\,w \geq \cfrac{S^1\,p\,d_2\,\eta}{\eta' \!+\! \gamma^{1} \!+\! \gamma^{2}}\,w.
		\label{eq19}
  \vspace{-0.2em}
	\end{equation}
	Moreover, by exploiting the fact that $R_0^{(2)} \geq d_2\,p\,\eta\,w/\eta'$, \eqref{eq19} holds only if the condition below is satisfied:\vspace{-0.2em}
	\begin{equation}  
		\frac{d_2\,p\,\eta}{\eta'}w-\frac{d_2\,p\,\eta\, \gamma^2}{\eta'(\eta' \!+\! \gamma^{1} \!+\! \gamma^{2})}\,w \geq \frac{S^1\,d_2\,p\,\eta}{\eta' \!+\! \gamma^{1} \!+\! \gamma^{2}}\,w.
		\label{eq20}
  \vspace{-0.2em}
	\end{equation}
After some algebraic manipulations, \eqref{eq20} reduces to $\eta' + \gamma^1 \geq S^1 \eta'$, which is always true. Note that the equality holds when $\gamma^1 = 0$ and thus, $S^1 = 1$.
	\item \textit{Case III}: By expressing \eqref{eq17} in terms of $\gamma^1$, we arrive at a third-order inequality of the form $a(\gamma^1)^3 + b(\gamma^1)^2 + c \gamma^1 + d < 0$, where the coefficients are specified as:\vspace{-0.2em}
	\begin{equation}
		\begin{cases}
			a &= R_0^{(1)} + R_0^{(2)} - 1, \vspace{0.2em} \\
			b &= \eta'(R_0^{(1)} + R_0^{(2)} - 1) + \gamma^2(3R_0^{(1)} + R_0^{(2)} - 3), \vspace{0.2em} \\
			c &= \gamma^2(R_0^{(1)} - 1) (2\eta' + 3\gamma^2 - \eta' R_0^{(2)}), \vspace{0.2em} \\
			d &= (\gamma^2)^2(R_0^{(1)} - 1) (\eta' + \gamma^2). \nonumber
		\end{cases}
  \vspace{-0.2em}
	\end{equation}
It is common knowledge that the product of the roots of the above third-degree equation is equal to $-a/d$ (i.e., $(R_0^{(1)} \!+\! R_0^{(2)} \!-\! 1)/\big((\gamma^2)^2(1 \!-\! R_0^{(1)}) (\eta' \!+\! \gamma^2)\big)$). When $R_0^{(1)} \!<\! 1$ and $R_0^{(1)} \!+\! R_0^{(2)} \!>\! 1$, this value is positive. This conclusion clearly implies that one of the three roots is a positive real number. Specifically, if $R_0^{(1)} \!<\! 1$ and $R_0^{(1)} \!+\! R_0^{(2)} \!>\! 1$, there exists some $\gamma^{1*}$ for which $a(\gamma^1)^3 + b(\gamma^1)^2 + c \gamma^1 + d < 0$ if $\gamma^1 < \gamma^{1*}$. As a result, the DFE is stable.
\end{itemize}
\vspace{-0.5em}

\section{Proof of Lemma~\ref{lemma2}}
\label{app_C}
\fontdimen2\font=0.50ex
With respect to the MF equations given in \eqref{eq:MFhomo}, the $\mathbf{F}$ matrix for the extended SCIR model can be written as:\vspace{-0.2em}
\begin{align}
	\mathbf{F} &=
		\begin{bmatrix}
			\kappa\, \beta_{\textrm{C}}\, \mathbf{F}_1 & \kappa\, \beta_{\textrm{I}}\, \mathbf{F}_1\\
			\bar{\kappa}\, \beta_{\textrm{C}}\, \mathbf{F}_1 & \bar{\kappa}\, \beta_{\textrm{I}}\, \mathbf{F}_1
		\end{bmatrix}. \nonumber 
  \vspace{-0.2em}
\end{align}
The matrix $\mathbf{V}$ is obtained in a manner similar to the original SCIR model given in Appendix~\ref{app_A}. The resulting outcome of the product $\mathbf{F \cdot V}^{-1}$ can be expressed as:\vspace{-0.2em}
\begin{align}
	\mathbf{F\cdot V}^{-1} &=
		\begin{bmatrix}
			\kappa\, \mathbf{L} & \kappa\, \mathbf{U}\\
			\bar{\kappa}\, \mathbf{L} & \bar{\kappa}\, \mathbf{U}
		\end{bmatrix}, \nonumber 
  \vspace{-0.2em}
\end{align}
where $\mathbf{L} \!=\! \mathbf{F}_1\cdot (\beta_{\textrm{C}}\,\mathbf{I} + \eta\,\beta_{\textrm{I}}\mathbf{V}_2^{-1})\cdot  \mathbf{V}_1^{-1}$ and $\mathbf{U} \!=\! \beta_{\textrm{I}} \mathbf{F}_1\cdot \mathbf{V}_2^{-1}$. Now, assuming that $\lambda$ and $\mathbf{\nu} \!=\! [\nu_1, \nu_2]^T$ are, respectively,~the eigenvalue and eigenvector of $\mathbf{F \cdot V}^{-1}$ (i.e., $\mathbf{F \cdot V}^{-1} \cdot \mathbf{\nu} \!=\! \lambda \mathbf{\nu}$), we have $\kappa\, \mathbf{L} \nu_1 \!+\! \kappa\,\mathbf{U} \nu_2 \!=\! \lambda\,\nu_1$ and $\bar{\kappa}\, \mathbf{L}\, \nu_1 \!+\! \bar{\kappa}\,\mathbf{U}\nu_2 \!=\! \lambda\,\nu_2$. From this, we conclude that $\bar{\kappa}\,\nu_1 \!=\! \kappa\,\nu_2$ and therefore, $(\kappa\,\mathbf{L} \!+\! \bar{\kappa}\mathbf{U})\nu_1 \!=\! \lambda\,\nu_1$. The latter equality implies that the eigenvalues of $\mathbf{F\cdot V}^{-1}$ are the same as those of $\kappa\,\mathbf{L} \!+\! \bar{\kappa}\,\mathbf{U}$ and hence, it is enough to compute $\rho(\kappa\,\mathbf{L}  + \bar{\kappa}\,\mathbf{U})$ in order to derive $\rho(\mathbf{F\cdot V}^{-1})$.\vspace{-0.5em}

\section{Proof of Proposition~\ref{prop2}}
\label{app_D}
\fontdimen2\font=0.50ex
We first need to show that the discriminant $\Delta_{\mathbf{H}} \!=\! \Tr(\mathbf{H})^2 - \det(\mathbf{H})$, where $\mathbf{H} \!\triangleq\! \kappa \mathbf{L} + \bar{\kappa}\mathbf{U}$, is positive. After some mathematical manipulations, $\Tr(\mathbf{H})$ and $\det(\mathbf{H})$ can be written~as:\vspace{-0.2em}
\begin{equation}
	\!\begin{cases}
		\!\!\Tr(\mathbf{H}) \!=\! \tilde{R}_0^{(1)\!} \!+\! S^2\! \overbrace{\Big(\!\tilde{R}_0^{(2)\!} - \cfrac{\kappa\,p\,d_2\,\eta\,\gamma^2}{\eta' (\eta' \!+\! \gamma^{1} \!+\! \gamma^{2})}\,w\Big)}^{z_1}, \vspace{0.2em}   \\
		\!\!\det(\mathbf{H}) \!=\! S^1S^2\tilde{R}_0^{(1)\!}\! \underbrace{\Big(\!\cfrac{\kappa\,p\,d_2\,\eta}{\eta' \!+\! \gamma^{1} \!+\! \gamma^{2}}\,w \!+\! \cfrac{\bar{\kappa}\,p\,d_2\,\beta_{\textrm{I}}}{\delta \!+\! \gamma^{1,\textrm{I}} \!+\! \gamma^{2,\textrm{I}}}\Big)}_{z_2}.  
	\end{cases}
 \vspace{-0.2em}
\end{equation}
Undertaking the same approach as in Proposition~\ref{prop1},~we can~represent $\Tr(\mathbf{H})$ and $\det(\mathbf{H})$ as $\tilde{R}_0^{(1)} \!+\! S^2\, z_1$ and $S^1\,S^2\tilde{R}_0^{(2)}\,z_2$,~respectively. Similarly, to show that $\Delta_{\mathbf{H}} \!\geq\! 0$, we need to prove that $z_1 \!\geq\! z_2$. From \eqref{eq:R02} and \eqref{eq11}, we observe that $R_0^{(2)} \geq p\,d_2\,\eta\,w/\eta'$. Therefore,\vspace{-0.2em}
\begin{equation}
	z_1 \geq \kappa\, \cfrac{p\,d_2\,\eta}{\eta'}\,w + \bar{\kappa}\, \cfrac{p\,d_2\,\beta_{\textrm{I}}}{\delta(\delta + \gamma^{1,\textrm{I}} + \gamma^{2,\textrm{I}})} \left(\gamma^{1,\textrm{I}} + \delta\right).
	\label{eq21}
 \vspace{-0.2em}
\end{equation}
To prove $z_1 \geq z_2$, it suffices to show that the right-hand side of \eqref{eq21} is not lesser than $z_2$. The inequality can be simplified to $\eta' \!+\! \gamma^{1} \!+\! \gamma^{2} \geq \eta' + \gamma^{2}$, which is always true. Following the same procedure as in Appendix~\ref{app_B}, it is straightforward to prove the stability of the DFE for the three cases in terms of $\tilde{R}_0^{(1)}$ and $\tilde{R}_0^{(2)}$ thresholds.\vspace{-0.4em}

\section{Proof of Lemma~\ref{lemma3}}
\label{app_E}
\fontdimen2\font=0.50ex
In the network with activity heterogeneity, there are $4N$~compartments related to infection, with $\mathbf{Q}$ consisting of rows corresponding to $\{C^1_i\}_{i=1}^{N}$, $\{C^2_i\}_{i=1}^{N}$, $\{I^1_i\}_{i=1}^{N}$, and $\{I^2_i\}_{i=1}^{N}$. Additionally, $\mathbf{F}$ and $\mathbf{V} \!=\! \mathbf{V}^{-} \!-\! \mathbf{V}^+$ are computed as described in Section~\ref{sec:org_scir}. Recall that $\mathbf{F}_{i,j} \!=\! \partial \mathcal{F}_i/\partial x_j$,  $\mathbf{V}^+_{i,j} \!=\! \partial \mathcal{V}^+_i/\partial x_j$, and  $\mathbf{V}^-_{i,j} \!=\! \partial \mathcal{V}^-_i/\partial x_j$ are all evaluated at $\mathcal{E}_0$. For example, the $(i,j)$-th element of the first block of $\mathbf{F}$ is calculated as $\kappa \frac{\partial}{\partial C^1_j}\sum_{j\in N} a_{i,j} S_i^1 \beta_{\textrm{C}} C^1_j=\beta_{\textrm{C}} S_i^1 a_{i,j}$, and as a result, the first block can be expressed as $\beta_{\textrm{C}} \Diag\big((S_i^1)_{i=1}^N\big) \cdot \mathbf{A}$. At the DFE,~this becomes $\beta_{\textrm{C}} \Diag\big(\boldsymbol{\hat{\gamma}}^2/(\boldsymbol{\hat{\gamma}}^1 + \boldsymbol{\hat{\gamma}}^2)\big) \cdot \mathbf{A}$.\vspace{-0.4em}

%
\bibliographystyle{IEEEtran}
\bibliography{IEEEabrv,myref}

\vfill

\end{document}